\documentclass{IEEEtran}

\usepackage{amsmath}
\usepackage{amssymb}
\usepackage{amsthm}
\usepackage{cite}
\usepackage{graphicx}
\usepackage{verbatim}
\usepackage{bbm}
\usepackage{subfigure}
\usepackage{url}
\usepackage{mathrsfs}
\usepackage{enumerate}
\usepackage{overpic}
\usepackage{booktabs}

\newcommand{\mb}[1]{{\boldsymbol #1}}

\renewcommand{\b}{\mb{b}}

\renewcommand{\d}{\mb{d}}
\newcommand{\e}{\mb{e}}

\renewcommand{\r}{\mb{r}}
\renewcommand{\u}{\mb{u}}
\newcommand{\bv}{\mb{v}}
\newcommand{\w}{\mb{w}}
\newcommand{\x}{\mb{x}}
\newcommand{\y}{\mb{y}}

\newcommand{\D}{\mb{D}}

\newcommand{\I}{\mb{I}}
\newcommand{\J}{\mb{J}}

\newcommand{\M}{\mb{M}}

\newcommand{\U}{\mb{U}}

\newcommand{\XX}{{\mathfrak X}}

\newcommand{\xmax}{|x_{\max}|}
\newcommand{\xmin}{|x_{\min}|}

\newcommand{\zero}{\mb{0}}

\newcommand{\beq}{\begin{equation}}
\newcommand{\eeq}{\end{equation}}

\newcommand{\hx}{{\widehat{\x}}}
\newcommand{\tx}{{\widetilde{\x}}}

\newcommand{\xor}{{\hx_{\mathrm{or}}}}
\newcommand{\xbth}{{\hx_{\mathrm{BTH}}}}
\newcommand{\xbomp}{{\hx_{\mathrm{BOMP}}}}
\newcommand{\xomp}{{\hx_{\mathrm{OMP}}}}

\newcommand{\eps}{\varepsilon}

\newcommand{\hS}{{\widehat{S}}}

\newcommand{\CC}{{\mathbb C}}
\newcommand{\RR}{{\mathbb R}}

\newcommand{\MSE}{{\mathrm{MSE}}}

\newcommand{\E}[1]{{\mathbb E} \! \left\{ #1 \right\}}
\newcommand{\Ex}[1]{{\mathbb E}_\x \! \left\{ #1 \right\}}
\renewcommand{\Pr}[1]{\mathrm{Pr} \kern -1pt \left\{ #1 \right\}}
\newcommand{\pd}[2]{\frac{\partial #1}{\partial #2}}

\newcommand{\Ra}[1]{{{\mathcal R}\! \left( #1 \right) }}

\newcommand{\pinv}{\dagger}

\newtheorem{theorem}{Theorem}
\newtheorem{lemma}{Lemma}

\newtheorem{corollary}{Corollary}

\DeclareMathOperator{\Tr}{Tr}

\DeclareMathOperator{\supp}{supp}
\DeclareMathOperator*{\argmax}{arg\,max}
\DeclareMathOperator*{\argmin}{arg\,min}

\begin{document}

\title{Near-Oracle Performance of Greedy Block-Sparse Estimation Techniques from Noisy Measurements}
\author{Zvika~Ben-Haim,~\IEEEmembership{Student~Member,~IEEE}, and~Yonina~C.~Eldar,~\IEEEmembership{Senior~Member,~IEEE}\thanks{%
The authors are with the Department of Electrical Engineering, Technion---Israel Institute of Technology, Haifa 32000, Israel (e-mail: \{zvikabh@tx, yonina@ee\}.technion.ac.il). This work was supported in part by a Magneton grant from the Israel Ministry of Industry and Trade, by the Israel Science Foundation under Grant 1081/07, and by the European Commission's FP7 Network of Excellence in Wireless COMmunications NEWCOM++ (grant agreement no.\ 216715).}}
\maketitle

\begin{abstract}
This paper examines the ability of greedy algorithms to estimate a block sparse parameter vector from noisy measurements. In particular, block sparse versions of the orthogonal matching pursuit and thresholding algorithms are analyzed under both adversarial and Gaussian noise models. In the adversarial setting, it is shown that estimation accuracy comes within a constant factor of the noise power. Under Gaussian noise, the Cram\'er--Rao bound is derived, and it is shown that the greedy techniques come close to this bound at high SNR\@. The guarantees are numerically compared with the actual performance of block and non-block algorithms, highlighting the advantages of block sparse techniques.
\end{abstract}

%%%%%%%%%%%%%%%%%%%%%%%%%%%%%%%%%%%%%%%%%%%%%%%%%%%%%%%%%%%%%%%%%%%%%%%%%%%%%%%%%%%%%%%%%%%%%%%%%%%%%%%%%%%%%%%%
\section{Introduction}
\label{se:intro}

The success of signal processing techniques depends to a large extent on the availability of an appropriate model which captures our knowledge of the system under consideration and translates it to a productive mathematical framework. There is consequently an ongoing search for mathematical models which can accurately describe real-world signals. In recent years, much research has been devoted to the sparse representation model, which stems from the observation that many signals can be approximated using a small number of elements, or ``atoms,'' chosen from a large dictionary \cite{donoho06b, candes06, bruckstein09}. Thus, we may write $\y = \D\x + \w$, where the signal $\y$ is a linear combination of a small number of columns of the dictionary matrix $\D$, corrupted by noise $\w$. Since only a small number of elements of $\D$ are required for this representation, the vector $\x$ is sparse, i.e., most of its entries equal $0$. It turns out that the sparsity assumption can be used to accurately estimate $\x$ from $\y$, even when the number of possible atoms (and thus, the length of $\x$) is greater than the number of measurements in $\y$ \cite{tropp06, donoho06, candes06}. This model has been used to great advantage in many fundamental fields of signal processing, including compressed sensing \cite{donoho06b, candes06}, denoising \cite{elad06}, deblurring \cite{bronstein05}, and interpolation \cite{fadili09}.

The assumption of sparsity is an example of a much more general class of signal models which can be described as a union of subspaces \cite{lu08b, blumensath09, eldar09}. Indeed, each support pattern defines a subspace of the space of possible parameter vectors. Saying that the parameter contains no more than $k$ nonzero entries is equivalent to stating that $\x$ belongs to the union of all such subspaces. Unions of subspaces are proving to be a powerful generalization of the sparsity model. Apart from ordinary sparsity, unions of subspaces have been applied to estimate signals as diverse as pulse streams \cite{gedalyahu09, gedalyahu11}, multi-band communications \cite{mishali09, mishali09b, mishali10}, and block sparse vectors \cite{yuan06, eldar09, stojnic09, eldar10}, the latter being the focus of this paper. The common thread running through these applications is the ability to exploit the union of subspaces structure in order to achieve accurate reconstruction of signals from a very low number of measurements.

The block sparsity model is based on the realization that in many practical sparse representation settings, not all support patterns are equally likely. Specifically, if a particular element of $\x$ is nonzero, then in many cases ``similar'' elements in $\x$ are also nonzero. The precise definition of similarity is context-dependent. For example, in Fourier-based dictionaries, neighboring frequency bins are often jointly nonzero, while in wavelet-based dictionaries, nonzero entries in a certain detail level are likely to be correlated with nonzeros in higher detail levels. Consequently, the sparsity model does not incorporate all of the structure present in the signal. The block sparsity approach aims to partially overcome this drawback by partitioning the vector $\x$ into blocks, each of which contains a small number of elements. The structure imposed by the block sparsity model is that no more than a small number $k$ of blocks are nonzero. The model thus favors the use of related atoms, rather than sporadic dictionary columns. Consequently, block sparsity is well-suited for those situations described above, in which specific atoms tend to be used together.

The usefulness of a model depends on the existence of efficient and effective methods for estimating a signal $\x$ from its measurements. Fortunately, estimators designed for the ordinary sparsity model can be readily adapted to the block sparse setting. Thus, previous work has described techniques such as block orthogonal matching pursuit (BOMP) \cite{eldar10} and the mixed $\ell_2/\ell_1$-optimization (L-OPT) \cite{eldar09, stojnic09}, the latter being a block version of the Lasso. In this paper, we also describe a block-sparse version of the thresholding algorithm, which we refer to as block-thresholding (BTH). The BOMP and BTH approaches are representatives of a class of so-called greedy algorithms, which attempt to identify the support of $\x$ by choosing at each step the most likely candidate. In this paper we restrict attention to these greedy techniques, which are simpler (and more naive) than convex relaxation techniques such as L-OPT, and are therefore more suitable for implementation in large-scale or computationally parsimonious settings.

Having described various estimation algorithms, it is natural to ask what can be guaranteed analytically about the performance of these methods in practice. For example, in the ordinary (non-block) sparsity setting, a rich collection of performance guarantees exists for various algorithms under different noise models. In particular, a distinction is made between adversarial and random noise models. In the former case, nothing is known about $\w$ except that it is bounded, $\|\w\|_2 \le \eps$; in particular, $\w$ might be chosen so as to maximally harm a given estimation algorithm. Consequently, guarantees in this case are relatively weak, ensuring only that the error in $\x$ is on the order of $\eps$ \cite{donoho06, candes06, tropp06}. By contrast, when the noise is random, estimation performance is considerably improved for most noise realizations \cite{tropp06, candes07, ben-haim10}.

It is natural to seek an extension of these results to the block sparsity model. In the absence of noise, successful recovery of a block sparse parameter $\x$ from measurements $\y = \D\x$ has been demonstrated in the past for both BOMP and L-OPT \cite{eldar09, eldar10}. However, to the best of our knowledge, the only result providing analytical guarantees for a block sparse estimator under noise was given in \cite{eldar09}, where the performance of L-OPT was analyzed under adversarial noise. The goal of this paper is to analyze the performance of the greedy algorithms BOMP and BTH under both adversarial and random noise models. As we will see, despite the fact that these greedy algorithms are simpler and more efficient to implement, their performance is close to the optimal achievable results.

Specifically, we first analyze the adversarial noise model, and show that both BOMP and BTH achieve an error on the order of $\eps$ when the noise is bounded by $\|\w\|_2 \le \eps$. These results generalize previous guarantees in several ways: First, when each block contains one element, we recover the non-block sparsity guarantee of Donoho et al.\ \cite{donoho06}. Second, when the noise bound $\eps$ equals $0$, we obtain the noise-free guarantees of Eldar et al.\ \cite{eldar10}.

We next turn to the random noise model, and examine in particular the case in which $\w$ is white Gaussian noise. We derive the Cram\'er--Rao bound (CRB) for estimating $\x$ from its measurements, and show that this bound equals the error of the ``oracle estimator'' which knows the locations of the nonzero blocks of $\x$. However, while the oracle estimator relies on information which is unavailable in practice, the CRB is known to be achievable by the maximum likelihood (ML) technique at high SNR\@. Unfortunately, the ML approach is NP-complete, and thus can probably not be implemented efficiently. Nevertheless, we proceed to show that both BOMP and BTH come within a nearly constant factor of the CRB at high SNR, for dictionaries satisfying suitable requirements. Once again, when each block contains one element, we can recover previously known guarantees for non-block sparsity \cite{ben-haim10} from our results. Furthermore, we show that in typical block sparse situations, the performance guarantees of block algorithms is substantially better than that of non-block techniques.

The rest of this paper is organized as follows. The block sparse setting is defined in Section~\ref{se:setting}, and the BOMP and BTH techniques are described in Section~\ref{se:algorithms}. The adversarial noise model is then analyzed in Section~\ref{se:adv}. The treatment of random noise begins with the derivation of the CRB in Section~\ref{se:crb}, while performance guarantees for this case appear in Section~\ref{se:gauss}. Finally, the guarantees and the CRB are compared with the actual performance of BOMP and BTH in a numerical study in Section~\ref{se:numer}.

%%%%%%%%%%%%%%%%%%%%%%%%%%%%%%%%%%%%%%%%%%%%%%%%%%%%%%%%%%%%%%%%%%%%%%%%%%%%%%%%%%%%%%%%%%%%%%%%%%%%%%%%%%%%%%%%
\section{Problem Setting}
\label{se:setting}

\subsection{Notation}

The following notation is used throughout the paper. Matrices and vectors are denoted by boldface uppercase letters $\M$ and boldface lowercase letters $\bv$, respectively. The $\ell_2$ norm of a vector $\bv$ is $\|\bv\|_2$ and the spectral norm of a matrix $\M$ is $\|\M\|$. The expectation of a random vector $\bv$ will be denoted $\E{\bv}$ or, occasionally, $\Ex{\bv}$, where the subscript is intended to emphasize the fact that the expectation is a function of the deterministic quantity $\x$. The adjoint and the Moore--Penrose pseudoinverse of a matrix $\M$ are denoted, respectively, by $\M^*$ and $\M^\pinv$, while the column space of $\M$ is $\Ra{\M}$. We denote by $\bv[i]$ the $i$th $d$-element block of a vector $\bv$ of length $N = Md$. Thus
\beq
\bv[i] \triangleq [v_{(i-1)d + 1}, v_{(i-1)d + 2}, \ldots, v_{id} ]^T, \quad 1 \le i \le M.
\eeq
Consequently, we may write
\beq
\bv = \big[\bv^T[1], \ldots, \bv^T[M] \big]^T.
\eeq
Similarly, given a matrix $\M$ having $N$ columns, the submatrix $\M[i]$ contains the columns $(i-1)d+1, (i-1)d+2, \ldots, id$ of $\M$, i.e., those columns of $\M$ which correspond to the $i$th block. The support $\supp(\bv)$ of $\bv$ is defined as the set of indices of nonzero blocks of $\bv$; formally
\beq
\supp(\bv) \triangleq \{ i : \bv[i] \ne \zero \}.
\eeq
Given an index set $I$, the vector $\bv_I$ is constructed as the subvector of $\bv$ containing the blocks indexed by $I$; in other words, if $I = \{ i_1, \ldots, i_p \}$, then
\beq
\bv_I = \big[ \bv^T[i_1], \ldots, \bv^T[i_p] \big]^T.
\eeq
Likewise, the submatrix $\M_I$ contains the column blocks indexed by $I$, so that
\beq
\M_I = \big[ \M[i_1], \ldots, \M[i_p] \big].
\eeq
To uniquely define $\bv_I$ and $\M_I$, we will assume as a convention that the elements of $I$ are sorted, i.e., $i_1 < i_2 < \cdots < i_p$.

\subsection{Problem Definition}

Let $\x \in \CC^N$ be a deterministic block-sparse vector, i.e., $\x$ consists of $M$ blocks $\x[1],\ldots,\x[M]$ of size $d$, of which at most $k$ are nonzero \cite{eldar10}. The maximum support size $k$ is assumed to be known. The block sparsity restriction can then be written as
\beq \label{eq:def XX}
\x \in \XX \triangleq \{ \bv \in \RR^N : |\supp(\bv)| \le k \}.
\eeq
For convenience, let $S \triangleq \supp(\x)$ be the support of the parameter $\x$, and let $s = |S|$. Note the distinction between $k$ and $s$: It is known that at most $k$ blocks are nonzero, but the actual number of nonzero blocks $s$ is unknown and may be smaller than $k$. In the sequel, it will be useful to define
\begin{align}
\xmax &\triangleq \max_{i \in S} \|\x[i]\|_2, \notag\\%&
\xmin &\triangleq \min_{i \in S} \|\x[i]\|_2.
\end{align}

The block sparse model differs from the more common non-block sparsity setting: in the latter, it is assumed that a small number of entries (rather than blocks) in the vector $\x$ are nonzero. To emphasize this difference, we will occasionally refer to the non-block sparsity model as ``ordinary'' or ``scalar'' sparsity. 

We are given noisy observations
\beq \label{eq:y=Dx+w}
\y = \D\x + \w
\eeq
where $\D \in \CC^{L \times N}$ is a known, deterministic dictionary, and $\w$ is a noise vector. Our goal is to estimate $\x$ from the measurements $\y$. It will be convenient to denote the $i$th column (or ``atom'') of $\D$ as $\d_i$. Thus we have
\beq
\D = [ \underbrace{\d_1, \ldots, \d_d}_{\D[1]}, \underbrace{\d_{d+1}, \ldots, \d_{2d}}_{\D[2]}, \ldots, \underbrace{\d_{N-d+1}, \ldots, \d_N}_{\D[M]} ].
\eeq
We assume for simplicity that the dictionary atoms are normalized, $\|\d_i\|_2 = 1$. We also assume that the measurement system is underdetermined, i.e., the number of measurements $L$ is less than the number of parameters $N$; thus, we must utilize the structure $\XX$, for otherwise we have no hope of recovering $\x$ from its measurements. Finally, we require that for any index set $I$ of size $|I| \le k$, the subdictionary $\D_I$ has full column rank. This latter assumption is needed to ensure that after a support set is chosen, one may estimate $\x$ using standard techniques for inverting an overcomplete set of linear equations, e.g., the least-squares approach.

We will provide performance guarantees for two separate noise models. First, we consider the adversarial setting, in which the noise is unknown but bounded,
\beq \label{eq:adv noise}
\|\w\|_2 \le \eps
\eeq
for a known constant $\eps > 0$. In this case the goal is to provide performance guarantees which hold for all values of $\w$ satisfying \eqref{eq:adv noise}. Second, we treat additive white Gaussian noise, in which
\beq \label{eq:AWGN}
\w \sim N(\zero,\sigma^2 \I).
\eeq
In this case $\w$ is unbounded, and the goal will be to provide guarantees which hold with high probability.

Following \cite{eldar10}, we define the block coherence of $\D$ as
\beq \label{eq:def mu B}
\mu_B \triangleq \max_{i \neq j} \frac{1}{d} \| \D^*[i] \D[j] \|.
\eeq
We also define the sub-coherence
\beq
\nu = \max_{1 \le \ell \le M} \, \, \max_{(\ell-1)d+1 \le i \ne j \le \ell d} |\d_i^* \d_j|.
\eeq
The block coherence and sub-coherence are generalizations of the concept of the coherence, which is defined as
\beq \label{eq:def mu}
\mu = \max_{1 \le i \ne j \le N} |\d_i^* \d_j|
\eeq
and applies to dictionaries regardless of whether they have a block structure.

%%%%%%%%%%%%%%%%%%%%%%%%%%%%%%%%%%%%%%%%%%%%%%%%%%%%%%%%%%%%%%%%%%%%%%%%%%%%%%%%%%%%%%%%%%%%%%%%%%%%%%%%%%%%%%%%
\section{Techniques for Block-Sparse Estimation}
\label{se:algorithms}

For reference and in order to fix notation, we now describe the two greedy algorithms for which we provide performance guarantees.

\paragraph{Block-Thresholding (BTH)}
We propose the following straightforward extension of the well-known thresholding algorithm. Given a measurement vector $\y \in \CC^L$, perform the following steps:
\begin{enumerate}
\item Compute the correlations
\beq
\rho_i = \| \D^*[i] \y \|_2, \quad i=1,\ldots, M.
\eeq
\item Find the $k$ largest correlations and denote their indices by $i_1, \ldots, i_k$. In other words, find a set of indices $\hS = \{i_1,\ldots,i_k\}$ such that $\rho_i \ge \rho_j$ for all $i \in \hS$ and $j \notin \hS$.
\item The reconstructed signal is given by
\beq \label{eq:def xbth}
\xbth = \argmin_{\tx: \supp(\tx) = \hS} \| \y - \D \tx \|_2.
\eeq
\end{enumerate}

\paragraph{Block Orthogonal Matching Pursuit (BOMP)}
The BOMP algorithm, based on the OMP algorithm \cite{pati93}, was first proposed in \cite{eldar10}.

Given a measurement vector $\y \in \CC^L$, perform the following steps:
\begin{enumerate}
\item Define $\r^0 = \y$.
\item For each $\ell = 1, \ldots, k$, do the following:
\begin{enumerate}
\item Set
\beq \label{eq:bomp i ell}
i_\ell = \argmax_i \|\D^*[i] \r^{\ell-1}\|_2.
\eeq
\item Set
\beq
\x^\ell = \argmin_{\tx: \supp(\tx) \subseteq \{ i_1, \ldots, i_\ell \}} \| \y - \D\tx \|_2.
\eeq
\item Set $\r^\ell = \y - \D\x^\ell$.
\end{enumerate}
\item The estimate is given by $\xbomp = \x^k$.
\end{enumerate}

\paragraph{Oracle Estimator}
We will find it useful to analyze the oracle estimator, which is defined as the least-squares solution within the true support set, i.e.,
\beq \label{eq:def xor}
\xor = \argmin_{\tx: \supp(\tx) \subseteq S} \|\x - \tx\|_2^2.
\eeq
Using the notation introduced above, we have
\begin{align} \label{eq:xor}
(\xor)_S     &= (\D_S^* \D_S)^{-1} \D_S^* \y, \notag\\
(\xor)_{S^C} &= \zero
\end{align}
where $S^C = \{ 1, \ldots, M \} \backslash S$ is the complement of the support set $S$. Note that the term ``oracle estimator'' is somewhat misleading, since $\xor$ relies on knowledge of the true support set $S$, and is therefore not a true estimator.

%%%%%%%%%%%%%%%%%%%%%%%%%%%%%%%%%%%%%%%%%%%%%%%%%%%%%%%%%%%%%%%%%%%%%%%%%%%%%%%%%%%%%%%%%%%%%%%%%%%%%%%%%%%%%%%%
\section{Guarantees for Adversarial Noise}
\label{se:adv}

We begin by stating our performance guarantees in the case of adversarial noise. The proofs of these results are quite technical and can be found in Appendix~\ref{ap:prf adv}.

\begin{theorem} \label{th:bth adv}
Consider the setting of Section~\ref{se:setting} with adversarial noise \eqref{eq:adv noise}. Suppose that
\beq \label{eq:bth adv cond}
(1 - (d-1)\nu)\xmin > 2 \eps \sqrt{1 + (d-1)\nu} + (2k-1) d \mu_B \xmax.
\eeq
Then, the BTH algorithm correctly identifies all elements of the support of $\x$, and its error is bounded by
\beq \label{eq:bth adv}
\|\xbth - \x\|_2^2 \le \frac{\eps^2}{1 - (d-1)\nu - (k-1)d \mu_B}.
\eeq
\end{theorem}

\begin{theorem} \label{th:bomp adv}
Consider the setting of Section~\ref{se:setting} with adversarial noise \eqref{eq:adv noise}. Suppose that
\beq \label{eq:bomp adv cond}
(1 - (d-1)\nu)\xmin > 2 \eps \sqrt{1 + (d-1)\nu} + (2k-1) d \mu_B \xmin.
\eeq
Then, the BOMP algorithm identifies all elements of $\supp(\x)$, and its error is bounded by
\beq \label{eq:bomp adv}
\|\xbomp - \x\|_2^2 \le \frac{\eps^2}{1 - (d-1)\nu - (k-1)d \mu_B}.
\eeq
\end{theorem}

The following remarks should be made concerning Theorems \ref{th:bth adv} and \ref{th:bomp adv}.

$\bullet$ \emph{Scalar sparsity:}
The scalar sparsity setting, in which $\x$ has no more than $k$ nonzero elements, can be recovered by choosing $d=1$. In this case, BOMP and BTH reduce to their scalar versions, which are called OMP and thresholding, respectively, and the block-coherence $\mu_B$ equals the coherence $\mu$ of \eqref{eq:def mu}. Theorems \ref{th:bth adv} and \ref{th:bomp adv} then coincide with the well-known results of Donoho et al.\ \cite{donoho06} for performance of scalar sparse signals under adversarial noise. As an example (and for future reference), the OMP performance guarantee is given below.

\begin{corollary}[Donoho et al.\ \cite{donoho06}]
Let $\y = \D\x + \w$ be a measurement vector of a signal $\x$ having sparsity $\|\x\|_0 \le k$. Suppose that the coherence $\mu$ of the dictionary $\D$ satisfies
\beq \label{eq:omp adv cond}
\xmin(1 - (2k-1)\mu) > 2\eps.
\eeq
Then, OMP recovers the correct support pattern of $\x$ and achieves an error bounded by
\beq \label{eq:omp adv}
\|\xomp - \x\|_2^2 \le \frac{\eps^2}{1 - (k-1)\mu}.
\eeq
\end{corollary}

Note that in the case of ordinary sparsity, $d=1$, and therefore $\xmin$ can be defined simply as the magnitude of the smallest nonzero element in $\x$.

$\bullet$ \emph{Benefits and limitations of block sparsity:}
It is interesting to compare the achievable performance guarantees when one utilizes the block-sparse structure, as opposed to merely using ordinary (scalar) sparsity information. For concreteness, we focus in this discussion on a comparison between OMP and BOMP, but identical conclusions can be drawn by comparing the thresholding algorithm with its block-sparse version BTH\@.
 
Consider a block sparse signal $\x$ as defined in Section~\ref{se:setting}. Such a signal can also be viewed as a scalar sparse signal of length $N = Md$, having no more than $sd$ nonzero elements. It is readily shown that the coherence $\mu$ satisfies $\nu \le \mu$ and $\mu_B \le \mu$ \cite{eldar10}. Consequently, 
\beq
\frac{\eps^2}{1 - (d-1)\nu - (k-1)d \mu_B} \le \frac{\eps^2}{1 - (sd-1)\mu}
\eeq
which implies that if the conditions for the performance guarantees of both BOMP and OMP hold, then the performance guarantee \eqref{eq:bomp adv} for BOMP will be at least as good as that of OMP \eqref{eq:omp adv}. Moreover, in typical block-sparse settings, both $\nu$ and $\mu_B$ will be substantially smaller than $\mu$ \cite{eldar10}, and the guarantees for BOMP will then be considerably better.

These results notwithstanding, it should be noted that BOMP should not automatically be preferred over OMP in any setting. This is because the condition \eqref{eq:bomp adv cond} of Theorem~\ref{th:bomp adv} can sometimes be weaker than that of OMP\@. Specifically, the factor $2 \eps \sqrt{1 + (d-1)\nu}$ in \eqref{eq:bomp adv cond} is larger than the analogous term $2 \eps$ in \eqref{eq:omp adv cond}.\footnote{The remaining terms in \eqref{eq:bomp adv cond} are always no worse than the corresponding terms in \eqref{eq:omp adv cond}.} This implies that if the sub-coherence $\nu$ is large, block sparse algorithms will not perform as well as their scalar counterparts. Such a result is to be expected: Highly correlated dictionary blocks may cause noise amplification, and in such cases, it may be preferable to separately correlate each atom with the measurements, rather than relying on the combined correlation of the entire block. Indeed, it would be quite surprising if a partition of \emph{any} dictionary $\D$ into arbitrary blocks could be shown to perform as well as a scalar sparsity algorithm, since the former adds a restriction on the possible support patterns of the vector $\x$. The lesson to be learned from this analysis is that block sparsity techniques are effective when the dictionary can be separated into blocks whose elements are orthogonal or nearly orthogonal.

$\bullet$ \emph{Noiseless case:}
The situation in which $\y = \D\x$, i.e., no noise is present in the system, has been previously analyzed in the context of block sparsity in \cite{eldar10}. This setting can be recovered by choosing the noise bound $\eps=0$. In this case, the condition \eqref{eq:bomp adv} simplifies to
\beq \label{eq:bomp noiseless cond}
(d-1)\nu + (2k-1)d \mu_B < 1
\eeq
and Theorem~\ref{th:bomp adv} then amounts to a guarantee for perfect recovery of $\x$ if \eqref{eq:bomp noiseless cond} holds. This result for the noise-free setting has been previously demonstrated in \cite[Thm.~3]{eldar10}.

Similarly, by substituting $\eps=0$ into Theorem~\ref{th:bth adv}, one obtains a perfect recovery condition for BTH in the noiseless setting. Specifically, if the condition
\beq \label{eq:bth noiseless cond}
(d-1) \nu \frac{\xmax}{\xmin} + (2k-1)d \mu_B < 1
\eeq
is satisfied, then BTH correctly recovers $\x$ from its noiseless measurements $\y = \D\x$.

Since BTH is a much simpler algorithm than BOMP, it is not surprising that the necessary condition \eqref{eq:bth noiseless cond} for BTH is somewhat stronger than the corresponding condition \eqref{eq:bomp noiseless cond} for BOMP\@. This difference between the conditions is indicative of the different strategies employed by the two techniques, and will be further discussed in Section~\ref{se:gauss}.

$\bullet$ \emph{Severity of the error:}
As in the scalar sparsity scenario, the presence of adversarial noise severely limits the ability of any algorithm to perform denoising. This is evident from Theorems \ref{th:bth adv} and \ref{th:bomp adv}, which guarantee only that the distance between the estimates and the true value of $\x$ is on the order of the noise magnitude $\eps$. Given our detailed knowledge of the structure of the signal $\x$, one would expect more powerful denoising capabilities for typical noise realizations. Consequently, in the remainder of this paper, we adopt the assumption of random noise, which cannot align itself so as to maximally interfere with the recovery algorithms.

%%%%%%%%%%%%%%%%%%%%%%%%%%%%%%%%%%%%%%%%%%%%%%%%%%%%%%%%%%%%%%%%%%%%%%%%%%%%%%%%%%%%%%%%%%%%%%%%%%%%%%%%%%%%%%%%
\section{The Cram\'er--Rao Bound}
\label{se:crb}

A central goal in assessing the quality of an estimator is to check its proximity to the best possible performance in the given setting. To this end, it is common practice to compute the CRB for unbiased estimators \cite{kay93}, i.e., those techniques $\hx$ for which the bias $\b(\x) \triangleq \Ex{\hx} - \x$ equals zero. The CRB is a lower bound on the mean-squared error $\MSE(\hx,\x) = \Ex{\|\hx-\x\|_2^2}$ for any unbiased estimator $\hx$.

To utilize the information inherent in the block sparsity structure, we apply the constrained CRB \cite{gorman90, StoicaNg98, ben-haim09, ben-haim09b} to the present setting. In the constrained estimation scenario, one often seeks estimators which are unbiased for all parameter values in the constraint set \cite{gorman90, StoicaNg98}. However, as we will see below, this requirement is too strict in the block sparse setting. Indeed, in Theorem~\ref{th:crb} we show that it is not possible to construct \emph{any} method which is unbiased for all feasible parameter values. Consequently, a weaker, local definition of unbiasedness is called for, which we refer to as $\XX$-unbiasedness \cite{ben-haim09b}.

Intuitively, an estimator $\hx$ is said to be $\XX$-unbiased at a point $\x \in \XX$ if $\Ex{\hx} = \x$ holds at the point $\x$ and at all points $\tx$ in $\XX$ which are sufficiently close to $\x$. To formally define $\XX$-unbiasedness, we first recall the concept of a feasible direction. A vector $\bv \in \CC^N$ is said to be a feasible direction at $\x$ if, for any sufficiently small $\alpha$, we have $\x + \alpha \bv \in \XX$. We then say that $\hx$ is $\XX$-unbiased at $\x$ if $\Ex{\hx} = \x$ and if
\beq \label{eq:bias req}
\left. \pd{\b(\x + \alpha \bv)}{\alpha} \right|_{\alpha=0} = 0
\eeq
for any feasible direction $\bv$. In other words, the bias is zero at $\x$ and remains unchanged, up to a first-order approximation, when moving away from $\x$ along feasible directions. This definition yields the following result, whose proof can be found in Appendix~\ref{ap:crb}.

\begin{theorem}[Cram\'er--Rao bound for block-sparse signals] \label{th:crb}
Consider the setting of Section~\ref{se:setting} in which the block sparse parameter vector $\x$ is to be estimated from measurements corrupted by Gaussian noise \eqref{eq:AWGN}.
\begin{enumerate}[(a)]
\item Suppose $\x$ contains fewer than $k$ nonzero blocks, i.e., $s<k$. Then, no finite-variance estimator is $\XX$-unbiased at $\x$.
\item Suppose $\x$ contains precisely $k$ nonzero blocks, i.e., $s=k$. Then, any estimator which is $\XX$-unbiased at $\x$ satisfies
\beq \label{eq:crb}
\MSE(\hx,\x) \ge \sigma^2 \Tr\left((\D_S^* \D_S)^{-1}\right).
\eeq
\end{enumerate}
\end{theorem}

We recall that both the MSE and the CRB are functions of the unknown vector $\x$, as is generally the case when estimating a deterministic parameter. It follows immediately from Theorem~\ref{th:crb} that no finite-variance estimator can satisfy $\Ex{\hx} = \x$ for all $\x \in \XX$, which explains why we previously avoided this simpler definition of unbiasedness in the constrained setting. Instead, restricting attention to a local unbiasedness requirement led to a finite CRB for almost all parameter values in $\x$: specifically, those parameters whose support is maximal, $|\supp(\x)| \triangleq s = k$.

For maximal-support values of $\x$, it is not difficult to show that the CRB \eqref{eq:crb} coincides with the MSE of the oracle estimator \eqref{eq:xor}. In this case it is possible to get a sense for the value of the bound, as follows. From \eqref{eq:norm Ds Ds} of Lemma~\ref{le:matrix norms} (see Appendix~\ref{ap:prf adv}), we have that none of the eigenvalues of $(\D_S^* \D_S)^{-1}$ are larger than $1/(1 - (d-1)\nu - (k-1)d\mu_B)$. Thus
\beq
\sigma^2 \Tr\left((\D_S^* \D_S)^{-1}\right) \le \frac{1}{1 - (d-1)\nu - (k-1)d\mu_B} k d \sigma^2.
\eeq
In other words, when the block coherence and sub-coherence of $\D$ are low, the bound of Theorem~\ref{th:crb} will be close to $k d \sigma^2$. This value is typically much lower than the total noise variance $\E{\|\w\|_2^2} = L \sigma^2$. Thus, at least according to the CRB, it is possible to achieve substantial denoising in the presence of random noise. This stands in contrast to the rather disappointing guarantees presented for adversarial noise in the previous section. We may thus hope that the performance will be improved when considering random noise.

As opposed to the oracle estimator, which cannot be implemented in practice, it is well-known that the CRB can be asymptotically achieved at high SNR by the maximum likelihood (ML) estimator \cite{kay93}. However, in the present setting, computing the ML estimator is NP-hard, and thus impractical. Consequently, it is of interest to determine whether there exist \emph{efficient} techniques which come close to the performance bound \eqref{eq:crb}, at least for high SNR values. As we will show in the next section, this question is answered in the affirmative: greedy block sparsity techniques do indeed approach the CRB for sufficiently high SNR\@.

%%%%%%%%%%%%%%%%%%%%%%%%%%%%%%%%%%%%%%%%%%%%%%%%%%%%%%%%%%%%%%%%%%%%%%%%%%%%%%%%%%%%%%%%%%%%%%%%%%%%%%%%%%%%%%%%
\section{Guarantees for Gaussian Noise}
\label{se:gauss}

In this section, we analyze the performance of block sparse algorithms when the noise $\w$ is a Gaussian random variable having mean zero and covariance $\sigma^2 \I$. Our main performance guarantees are summarized in Theorems \ref{th:bth gauss} and \ref{th:bomp gauss}. The proofs of these theorems are found in Appendix~\ref{ap:prf gauss}.

\begin{theorem} \label{th:bth gauss}
Consider the setting of Section~\ref{se:setting} with additive white Gaussian noise $\w \sim N(\zero, \sigma^2 \I)$. Suppose it is known that
\begin{align} \label{eq:bth gauss cond}
&(1 - (d-1)\nu)\xmin - (2k-1)d \mu_B \xmax \notag\\
&\hspace{15mm} \ge 2 \sigma \sqrt{2 \alpha d (1 + (d-1)\nu) \log N }
\end{align}
for some constant $\alpha \ge 1/(2d\log N)$. Then, with probability exceeding
\beq \label{eq:bth gauss prob}
1 - \frac{0.8 d ( 2 \alpha d \log N )^{d/2-1} }{N^{\alpha d-1}}
\eeq
the BTH algorithm identifies the correct support of $\x$ and achieves an error bounded by
\beq \label{eq:bth gauss perf}
\|\xbth - \x\|_2^2 \le \frac{2 \alpha (1 + (d-1)\nu)}{(1 - (d-1)\nu - (k-1)d \mu_B)^2} d k \sigma^2 \log N.
\eeq
\end{theorem}

\begin{theorem} \label{th:bomp gauss}
Consider the setting of Section~\ref{se:setting} with additive white Gaussian noise $\w \sim N(\zero, \sigma^2 \I)$. Suppose it is known that
\begin{align} \label{eq:bomp gauss cond}
&(1 - (d-1)\nu)\xmin - (2k-1)d \mu_B \xmin \notag\\
&\hspace{15mm} \ge 2 \sigma \sqrt{2\alpha d (1 + (d-1)\nu) \log N }
\end{align}
for some constant $\alpha \ge 1/(2d\log N)$. Then, with probability exceeding \eqref{eq:bth gauss prob}, the BOMP algorithm identifies the correct support of $\x$ and achieves an error bounded by
\beq \label{eq:bomp gauss perf}
\|\xbomp - \x\|_2^2 \le \frac{2 \alpha (1 + (d-1)\nu)}{(1 - (d-1)\nu - (k-1)d \mu_B)^2} d k \sigma^2 \log N.
\eeq
\end{theorem}

We now provide some insights into the performance of block-sparse algorithms under random noise.

$\bullet$ \emph{Random noise vs.\ adversarial noise:}
As noted in Section~\ref{se:adv}, performance guarantees in the case of adversarial noise can ensure a recovery error on the order of the total noise magnitude. This is a result of the fact that the noise could, in principle, be concentrated in a single nonzero component of $\x$, whereupon it would be indistinguishable from the signal. However, for random noise, such an event is highly unlikely. Consequently, Theorems \ref{th:bth gauss} and \ref{th:bomp gauss} provide much tighter performance guarantees: both theorems demonstrate that, with high probability, the estimation error is on the order of $d k \sigma^2 \log N$, i.e., within a constant times $\log N$ of the CRB presented in Section~\ref{se:crb}. Since the noise variance $\E{\|\w\|^2}$ is given by $N \sigma^2$, and since typically $d k \log N \ll N$, we conclude that the block sparse algorithms have successfully removed a large portion of the noise, owing to the utilization of the union-of-subspaces structure.

$\bullet$ \emph{BOMP vs.\ BTH: }
Comparing Theorems \ref{th:bth gauss} and \ref{th:bomp gauss} leads to an important insight concerning the advantage of the more sophisticated BOMP algorithm over its simpler counterpart. Indeed, the guarantee for BOMP requires condition \eqref{eq:bomp gauss cond}, which basically states that $\xmin$ must be larger than a constant multiplied by the standard deviation of the noise. By contrast, for the BTH guarantee one requires the stronger condition \eqref{eq:bth gauss cond}, which can be interpreted as requiring $\xmin$ to be larger than a small constant times $\xmax$, plus another constant times the noise standard deviation.

To explain this difference, recall from Section~\ref{se:algorithms} that the BTH approach relies on a single support-identification stage in which the blocks most highly correlated with the measurements are chosen as the estimated support set $\hS$. Thus, for BTH to correctly identify the support, each block in $S$ must be sufficiently large in magnitude to overcome interference from the noise and from the remaining blocks. Condition \eqref{eq:bth gauss cond} can therefore be interpreted as a requirement that the magnitude $\xmin$ of the smallest nonzero block must be larger than the sum of the interference from the large nonzero blocks (the $\xmax$ term) and the noise. By contrast, the BOMP algorithm iteratively identifies support elements, maintaining a residual vector $\r^\ell$ containing the components of the measurement vector which have yet to be identified. Thus, BOMP requires only the ability to separately isolate each nonzero block, and hence its weaker condition \eqref{eq:bomp gauss cond}, which necessitates only that $\xmin$ be larger than the noise.

Finally, it should be noted that when BTH and BOMP both identify the correct support set, the estimates of the two algorithms coincide, explaining the identical bounds on their performance. The conclusion from this analysis is that BOMP should be preferred if a wide dynamic range of block magnitudes is possible, but that when all blocks have roughly the same size, the simpler and more efficient BTH technique can be used.

\begin{table*}
% these and a few other results are stored in 100726\results.xls.
% they were computed using 100726\multi_compare.m.
% see also notebook p.672.
\begin{center}
\begin{tabular}{rrrrrrrrrrr}
\toprule
\multicolumn{4}{c}{Problem Dimensions} & \multicolumn{2}{c}{Coherence} & \multicolumn{2}{c}{OMP} & \multicolumn{2}{c}{Block-OMP} & \multicolumn{1}{c}{Cram\'er--Rao} \cr
\cmidrule(r){1-4}
\cmidrule(r){5-6}
\cmidrule(r){7-8}
\cmidrule(r){9-10}
\cmidrule(r){11-11}
Blocks & Block size & Measurements & Sparsity \cr
$M$ & $d$ & $L$ & $k$ & $\mu$ & $\mu_B$ & Guarantee/$\sigma^2$  & $\sigma_{\max}$ & Guarantee/$\sigma^2$  & $\sigma_{\max}$ & CRB/$\sigma^2$\cr
\midrule
1200   & 5   & 3000 & 1   & 0.10  & 0.026   & 301.0  & 0.033   &  37.0  & 0.160 &  5.0 \\
1200   & 5   & 3000 & 2   & 0.10  & 0.026   & ---    & ---     &  98.8  & 0.110 & 10.0 \\
1200   & 5   & 3000 & 3   & 0.10  & 0.026   & ---    & ---     & 204.4  & 0.063 & 15.1 \\
1200   & 5   & 3000 & 4   & 0.10  & 0.026   & ---    & ---     & 417.0  & 0.010 & 20.1 \\
1200   & 5   & 3000 & 5   & 0.10  & 0.026   & ---    & ---     & ---    & ---   & 25.2 \\
\midrule
1200   & 5   & 3000 & 3   & 0.10  & 0.026   & ---    & ---     & 204.4  & 0.063 & 15.1 \\
 600   & 10  & 3000 & 3   & 0.10  & 0.015   & ---    & ---     & 364.3  & 0.049 & 30.2 \\
 300   & 20  & 3000 & 3   & 0.10  & 0.010   & ---    & ---     & 879.1  & 0.008 & 60.8 \\
 200   & 30  & 3000 & 3   & 0.10  & 0.007   & ---    & ---     & ---    & ---   & 91.8 \\
\midrule
1200   & 5   & 3000 & 1   & 0.10  & 0.026   & 301.0  & 0.033   &  37.0  & 0.160 &  5.0 \\
1200   & 5   & 1000 & 1   & 0.17  & 0.043   & ---    & ---     &  37.0  & 0.144 &  5.0 \\
1200   & 5   &  500 & 1   & 0.25  & 0.060   & ---    & ---     &  37.0  & 0.128 &  5.0 \\
1200   & 5   &  100 & 1   & 0.51  & 0.133   & ---    & ---     &  37.0  & 0.062 &  5.0 \\
1200   & 5   &   50 & 1   & 0.71  & 0.165   & ---    & ---     &  37.0  & 0.032 &  5.0 \\
1200   & 5   &   20 & 1   & 0.90  & 0.197   & ---    & ---     &  37.0  & 0.003 &  5.0 \\
1200   & 5   &   10 & 1   & 0.98  & 0.200   & ---    & ---     & ---    & ---   &  5.0 \\
\bottomrule
\end{tabular}
\end{center}
\caption{Performance Guarantees for OMP and Block-OMP}
\label{ta:1}
\end{table*}

$\bullet$ \emph{Scalar sparsity:} It is interesting to note that known results for scalar sparsity algorithms can be recovered from our block sparsity guarantees, by substituting $d=1$ into Theorems \ref{th:bth gauss} and \ref{th:bomp gauss}. For example, consider the BOMP guarantee (Theorem~\ref{th:bomp gauss}). In the scalar case, this algorithm is known as OMP, and its performance guarantee can be written as follows.

\begin{corollary} \label{co:omp gauss}
Let $\y = \D\x + \w$ be a measurement vector of a signal $\x$ having sparsity $\|\x\|_0 \le k$. Suppose the coherence $\mu$ of $\D$ satisfies
\beq
\xmin(1 - (2k-1)\mu) \ge 2 \sigma \sqrt{2 \alpha \log N}
\eeq
for some $\alpha>1$. Then, with probability exceeding
\beq\label{eq:omp gauss prob}
1 - \frac{0.8/\sqrt{2}}{N^{\alpha-1} \sqrt{\alpha \log N}}
\eeq
the OMP algorithm recovers the correct support of $\x$, and achieves an error bounded by
\beq
\|\xomp - \x\|_2^2 \le \frac{2\alpha}{(1 - (k-1)\mu)^2} k \sigma^2 \log N.
\eeq
\end{corollary}

Corollary~\ref{co:omp gauss} is nearly identical to \cite[Thm.~4]{ben-haim10}, with the only difference being that the constant $0.8/\sqrt{2} \approx 0.566$ in \eqref{eq:omp gauss prob} is replaced in \cite{ben-haim10} with the slightly better constant $1/\sqrt{\pi} \approx 0.564$. This slight discrepancy can be resolved if the more accurate version \eqref{eq:le:chi square gamma} of Lemma~\ref{le:chi square} is used in the proof of Theorem~\ref{th:bomp gauss}, but the resulting expression becomes much more cumbersome in the block sparse case.

$\bullet$ \emph{Block sparsity vs.\ scalar sparsity:} A legitimate question is whether the incorporation of the block sparsity structure substantially assists estimation algorithms. In other words, do the performance guarantees of the block algorithms BOMP and BTH compare favorably with the results achievable on identical signals using scalar sparsity algorithms, such as OMP and thresholding? This question is examined numerically in the next section.

%%%%%%%%%%%%%%%%%%%%%%%%%%%%%%%%%%%%%%%%%%%%%%%%%%%%%%%%%%%%%%%%%%%%%%%%%%%%%%%%%%%%%%%%%%%%%%%%%%%%%%%%%%%%%%%%
\section{Numerical Experiments}
\label{se:numer}

\begin{figure*}
% To calculate the results: 100802\calculate.m (r820)
% Results saved in results2.mat
% To plot the results: 100802\plotme (r820). Make sure that results2.mat is loaded.
% Each figure should be resized to 8x6 cm and then saved as EPS.
% In Matlab, Fig.1=OMP, Fig.2=BTH, Fig.3=OMP, Fig.4=Thr.
\centerline{%
\subfigure[Block-OMP]{%
\includegraphics{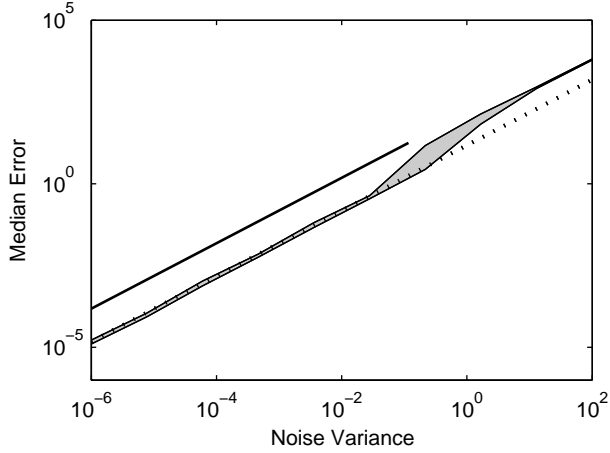}%
\label{fi:1bomp}}
\hfil
\subfigure[Block-Thresholding]{%
\includegraphics{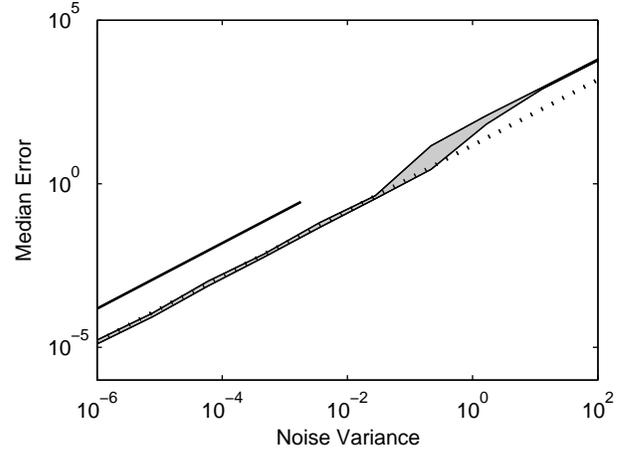}%
\label{fi:1bth}}%
}
\centerline{%
\subfigure[OMP]{%
\includegraphics{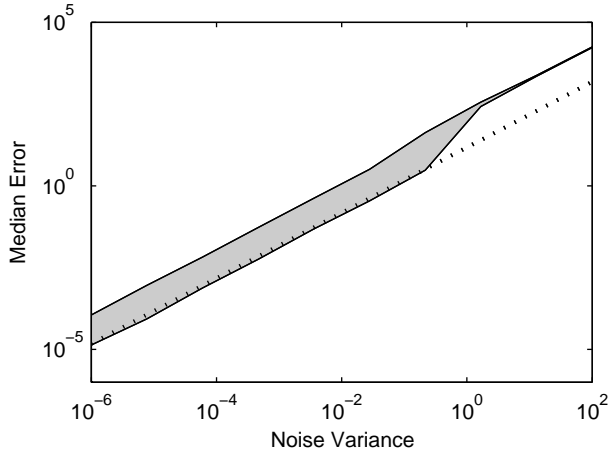}%
\label{fi:1omp}}
\hfil
\subfigure[Thresholding]{%
\includegraphics{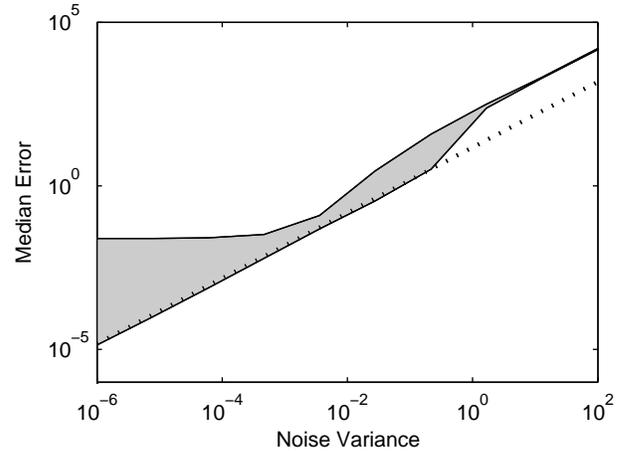}%
\label{fi:1thr}}%
}
\caption{Median squared error as a function of the noise variance for block and scalar sparse estimation algorithms. The shaded region indicates the range of errors encountered for different parameter values. The dotted line plots the CRB. The thick solid line in Figs. \ref{fi:1bomp} and \ref{fi:1bth} indicates the performance guarantees for the block sparse algorithms; no guarantee can be made for the scalar sparsity techniques in Figs. \ref{fi:1omp} and \ref{fi:1thr}.}
\label{fi:1}
\end{figure*}

From a practical point of view, it is important to determine whether the use of block sparse algorithms contributes significantly to the performance of estimation algorithms. After all, any block sparse signal containing $k$ nonzero blocks of size $d$ can also be viewed as a sparse signal containing $k d$ nonzero elements. Is there a significant benefit in using the block algorithms rather than the ordinary scalar versions?

There are two possible approaches to answering this question. First, one may compare the performance achieved in practice by block sparse and scalar sparse algorithms. This requires a complete specification of the problem setting, including a choice of the parameter value $\x$, which is unknown in practice. Alternatively, one can compare the performance guarantees for block sparse techniques, which were derived in Section~\ref{se:gauss}, to the previously known guarantees for scalar approaches \cite{ben-haim10b}. The performance guarantees apply to all parameter values having a specified sparsity level, and are therefore more general. However, there may be a gap between the guarantee and the performance observed in practice. In order to take advantage of both approaches, in the following we compare both the actual performance and the guarantees of the various algorithms discussed in this paper.

In our experiments, we used dictionaries containing orthonormal blocks. Such dictionaries were constructed by first generating a random $L \times N$ matrix containing IID, zero-mean Gaussian random variables, and then performing a Gram--Schmidt procedure separately on the columns of each block. As a first experiment, we generated a variety of such dictionaries, and computed their coherence $\mu$ and block coherence $\mu_B$. (The sub-coherence of dictionaries generated in this manner is necessarily $\nu=0$.) These values were used to compute performance guarantees for BOMP (using Theorem~\ref{th:bomp gauss}) and for OMP (using Corollary~\ref{co:omp gauss}). We assumed throughout that the minimum norm $\xmin$ among nonzero blocks equals $1$ and that the minimum nonzero element equals $1/\sqrt{d}$. Some typical results are listed in Table~\ref{ta:1}. To compute the guarantees in this table, the smallest value of $\alpha$ yielding a 99\% probability of success was chosen. The resulting guarantee is listed in multiples of $\sigma^2$. For example, a value of $\text{Guarantee}/\sigma^2 = 100$ means that $\|\hx-\x\|_2^2 \le 100\sigma^2$ for 99\% of the noise realizations. Also listed in Table~\ref{ta:1} are the maximum noise standard deviations $\sigma_{\max}$ for which the performance guarantees still hold. A dash (---) indicates that no guarantee can be made for the given setting even in the noise-free case.

It is evident from Table~\ref{ta:1} that the block sparse algorithm BOMP is guaranteed to perform over a much wider range of problem settings than the scalar OMP approach. Furthermore, even when performance guarantees are provided for both techniques, those for BOMP are substantially stronger. To provide merely one striking example from Table~\ref{ta:1}, note that $50$ measurements suffice for BOMP to identify a signal composed of a single $5$-element block among a set of $1200$ possible blocks, whereas for OMP to identify such a signal at the same noise level, as many as $3000$ measurements are required. The reason for this advantage is clear: the OMP algorithm must separately identify each nonzero component of the signal, and must therefore choose among a total of $\binom{1200}{5} \approx 2.1 \cdot 10^{13}$ possible support sets. This is obviously more challenging than identifying one nonzero block among a set of $1200$ possibilities. Clearly, then, knowledge of a block-sparse structure can substantially improve performance if it is correctly utilized.

Table~\ref{ta:1} also compares the performance guarantees with the CRB of Theorem~\ref{th:crb}. The CRB is listed for a random choice of support set $S$ containing precisely $k$ nonzero blocks; however, choosing different sets $S$ only has a small effect on the value of the bound. The gap between these lower and upper bounds is not inconsiderable, and is typically on the order of a factor of 10\@. There are several reasons for this gap. First, the performance guarantees plotted above indicate an error which is obtained with 99\% confidence, whereas the CRB is a bound on the MSE\@. By its very nature, the MSE averages out unusually disruptive noise realizations, and thus tends to be more optimistic. Second, different values of $\x$ may yield significantly different performance; the performance guarantees apply to \emph{all} values of $\x$, whereas the CRB is plotted for a single, typical parameter value. Third, some loss of tightness undoubtedly results from the derivations of the theorems, i.e., there may still be room for improved bounds.

To measure the relative influence of these factors, we performed another experiment, in which the guarantees were compared with the actual performance of the various algorithms. To overcome the aforementioned pessimistic effect of a guarantee which holds with overwhelming probability, in this second experiment we computed guarantees with a 50\% confidence level. In other words, these are assurances on the median of the distance between $\x$ and its estimate, which captures the typical estimation error. We also computed the actual median error of the various algorithms for a variety of parameter values.

\begin{figure*}
% To calculate the results: 100802\calculate.m (r818)
% Results saved in results.mat (r818)
% To plot the results: 100802\plotme (r820). Make sure that results.mat is loaded.
% Each figure should be resized to 8x6 cm and then saved as EPS.
% In Matlab, Fig.1=OMP, Fig.2=BTH, Fig.3=OMP (not used here), Fig.4=Thr (not used here).
\centerline{%
\subfigure[Block-OMP]{%
\includegraphics{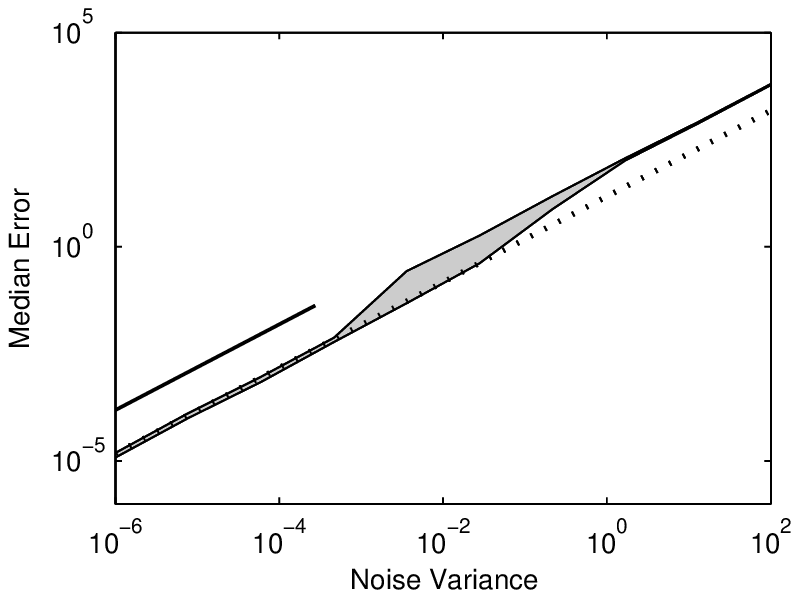}%
\label{fi:2bomp}}
\hfil
\subfigure[Block-Thresholding]{%
\includegraphics{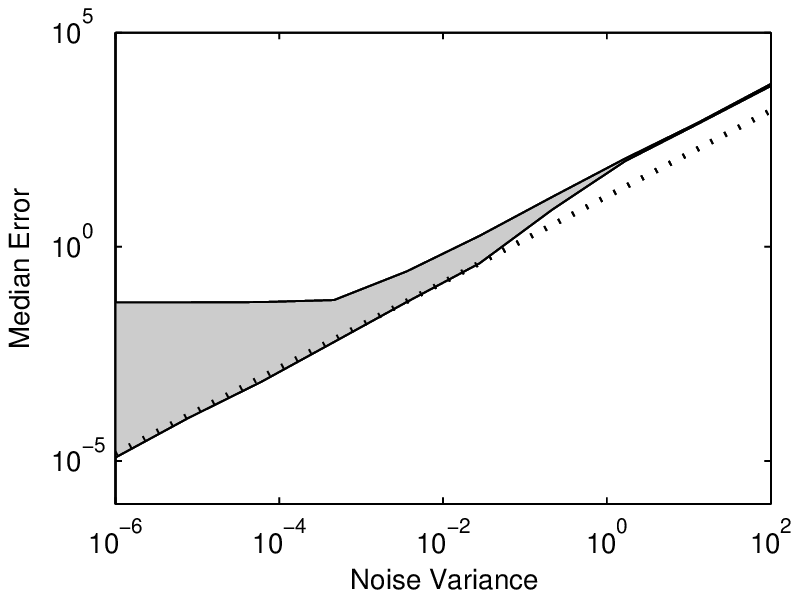}%
\label{fi:2bth}}%
}
\caption{Median squared error as a function of the noise variance for block sparse estimation algorithms. The shaded region indicates the range of errors encountered for different parameter values. The dotted line plots the CRB. The thick solid line in Fig. \ref{fi:2bomp} indicates the performance guarantee for BOMP; no guarantee can be made for BTH. The deteriorated performance of BTH is a result of the existence of low-magnitude blocks.}
\label{fi:2}
\end{figure*}

The details of this experiment are as follows. We constructed a $3000 \times 6000$ dictionary $\D$ containing $M=1200$ blocks of $d=5$ atoms each, using the orthogonalization algorithm described above. The resulting coherence of $\D$ was $\mu = 0.094$, the block coherence was $\mu_B = 0.026$, and since each block was orthonormal, the sub-coherence was $\nu = 0$. We then constructed a variety of block sparse vectors $\x$, each having $s=3$ nonzero blocks, with $\xmin = 2\sqrt{d}$ and $\xmax = 3\sqrt{d}$. We chose the parameter vectors so as to cover as wide a range of scenarios as possible, within the aforementioned requirements. For example, some parameter vectors contained a block with a single nonzero component whose value was $\xmax$, while other vectors contained a block with each of the $d$ elements receiving a value of $\xmax/\sqrt{d}$. Although it is clearly not feasible to cover the full range of possible parameter vectors, it is hoped that in this way some sense is given of the variability in performance for different parameter values. Indeed, as shown below, different parameters often yield widely differing estimation errors.

For each choice of a parameter vector, $20$ noise realizations were generated and the resulting measurement vector $\y$ was computed using \eqref{eq:y=Dx+w}. The BOMP, BTH, OMP, and thresholding algorithms were then applied to each of the measurement vectors. For every technique and each parameter vector, the median estimation error (among the noise realizations) was computed. The range of median estimation errors obtained for different choices of $\x$ is plotted as a shaded area in Fig.~\ref{fi:1}.

In the present setting, neither of the scalar sparsity algorithms was capable of providing a performance guarantee. For BOMP and BTH, performance guarantees were available, and these are plotted as a solid line in Fig.~\ref{fi:1}. These guarantees are valid only up to a certain maximal noise variance, at which point the solid line in Fig.~\ref{fi:1} stops. The results are also compared with the CRB of Theorem~\ref{th:crb}. It should be emphasized that the CRB is a bound on the MSE, rather than the median error, although in practice the differences between these two quantities appear to be quite small. It is also worth recalling that the CRB is a bound on unbiased estimators, while all of the techniques discussed herein are biased; nevertheless, it is evident that the CRB still provides a rough measure of the optimal performance of the proposed algorithms.

Several comments are in order concerning Fig.~\ref{fi:1}. First, the performance of both block sparse algorithms exhibits a transition: near-CRB performance for low noise levels deteriorates substantially when the noise level crosses a certain threshold. This behavior qualitatively matches the predictions of the performance guarantees, which ensure support recovery and near-CRB performance for sufficiently low noise levels. The threshold at which this transition occurs is identified fairly accurately for BOMP, and less so for BTH, although it is possible that there exist some (untested) parameter values for which the BTH transition occurs at lower noise levels. However, the numeric value of the performance guarantee is somewhat pessimistic: while the observed performance is close to the CRB for all parameter values, analytically one can guarantee only that the median error will not be larger than approximately 10 times the CRB\@. This result is most likely due to the various inequalities employed in the proofs of Theorems \ref{th:bth gauss} and \ref{th:bomp gauss}. Indeed, since the correct support is identified with high probability for most noise realizations, the BTH and BOMP algorithms will likely tend to coincide with the oracle estimator, whose error equals that of the CRB\@. The question of formally proving such a claim remains a topic for further research.

The advantages of the block sparse approach become evident when compared with scalar sparsity algorithms (Figs. \ref{fi:1omp} and \ref{fi:1thr}). For the scalar techniques, no performance guarantees can be made in the present setting. Unlike the block sparsity algorithms, the scalar approaches fail to recover the correct parameter vector even when the noise is negligible, and for some parameter values, their error does not converge to the CRB\@. The thresholding algorithm, in particular, ceases to improve (for some parameter values) as the noise is reduced, while the OMP approach, although significantly better than thresholding, does not converge to the CRB as do the block sparse techniques. This demonstrates the advantages of utilizing the fact that the signal is known to have a block-sparse structure.

The performance of BOMP (Fig.~\ref{fi:1bomp}) is quite similar to that of BTH (Fig.~\ref{fi:1bth}) in the experiment above. This is not surprising when one compares our problem setting with the guarantees of Section~\ref{se:gauss}. Indeed, as we have seen, the primary difference between the BOMP and BTH algorithms is that the one-shot support estimation employed by BTH causes large-magnitude blocks to overshadow small-magnitude nonzero blocks. In the setting of Fig.~\ref{fi:1}, the range of magnitudes between $\xmax = 3\sqrt{d}$ and $\xmin = 2\sqrt{d}$ is not very large, and therefore BTH performs nearly as well as BOMP\@. The advantages of BOMP become readily apparent if one considers a wider dynamic range. This is illustrated in Fig.~\ref{fi:2}, in which the setup is identical to that of the previous experiment, except that parameter vectors having $\xmin = 0.1 \sqrt{d}$ and $\xmax = \sqrt{d}$ were chosen, yielding a 10-fold dynamic range in the block magnitudes. In this case, while the guarantee for BOMP is hardly changed, the conditions for Theorem~\ref{th:bth gauss} no longer hold, so that nothing can be ensured concerning the BTH technique. Indeed, in Fig.~\ref{fi:2} we see that BTH performs poorly for some parameter values even when the noise level is low, and its performance is no longer proportional to the CRB\@.

%%%%%%%%%%%%%%%%%%%%%%%%%%%%%%%%%%%%%%%%%%%%%%%%%%%%%%%%%%%%%%%%%%%%%%%%%%%%%%%%%%%%%%%%%%%%%%%%%%%%%%%%%%%%%%%%
\section{Conclusion}

In this paper, we analyzed the performance of the greedy block algorithms BOMP and BTH under the adversarial and Gaussian noise models. In the adversarial setting $\|\w\|_2 \le \eps$, we showed that the estimation error equals a constant times the noise bound $\eps$, which shows that performance in this case will not necessarily reduce the noise power. The situation is much better in the presence of random noise, where we saw that, under suitable conditions, greedy techniques obtain an error on the order of $d k \sigma^2 \log N$ with high probability; this is substantially lower than the input noise power $N \sigma^2$. Indeed, the BTH and BOMP algorithms come close to the CRB and the error of the oracle estimator.

There remain many open questions concerning the performance of block sparse techniques under random noise. For example, for scalar sparsity, performance guarantees for convex relaxation techniques do not require assumptions on the SNR. An important challenge is to determine whether similar SNR-independent results can be demonstrated for block convex relaxation techniques such as L-OPT. Furthermore, it is well-known that scalar sparsity guarantees can be strengthened if the restricted isometry constants of the dictionary $\D$ are known, as is the case, for example, when $\D$ is chosen from an appropriate random ensemble. Thus, it is also of interest to provide guarantees for block techniques under random noise based on an extension of the RIP to the block sparse setting. One such extension has already been proposed in \cite{eldar09}, and its application to the Gaussian noise model may provide tighter bounds for some performance algorithms.

%%%%%%%%%%%%%%%%%%%%%%%%%%%%%%%%%%%%%%%%%%%%%%%%%%%%%%%%%%%%%%%%%%%%%%%%%%%%%%%%%%%%%%%%%%%%%%%%%%%%%%%%%%%%%%%%
\appendices
\section{Proofs for Adversarial Noise}
\label{ap:prf adv}

We begin by providing several lemmas which will prove useful for the analysis under both the adversarial and the Gaussian noise models.

\begin{lemma} \label{le:matrix norms}
Given a dictionary $\D$ having block coherence $\mu_B$ and sub-coherence $\nu$, we have
\beq
\|\D^*[i]\D[j]\|               \le d \mu_B \quad \text{for all $i \ne j$} \label{eq:norm Di Dj}
\eeq
and
\beq
\|\D[i]\|^2 = \|\D^*[i]\D[i]\| \le 1 + (d-1)\nu .                         \label{eq:norm Di Di}
\eeq
If $1 - (d-1)\nu > 0$, then
\beq
\|(\D^*[i]\D[i])^{-1}\|        \le \frac{1}{1 - (d-1)\nu}.                \label{eq:norm inv Di Di}
\eeq
Suppose $1 - (d-1)\nu - (k-1)d\mu_B > 0$ and let $I$ be an index set with $|I| \le k$. Then
\beq
\|(\D_I^* \D_I)^{-1}\|         \le \frac{1}{1 - (d-1)\nu - (k-1)d\mu_B}.  \label{eq:norm Ds Ds}
\eeq
\end{lemma}

\begin{proof}
The bound \eqref{eq:norm Di Dj} follows directly from the definition \eqref{eq:def mu B} of block coherence. To prove \eqref{eq:norm Di Di}--\eqref{eq:norm inv Di Di}, observe that the diagonal elements of the matrix $\D^*[i] \D[i]$ equal $1$, while the off-diagonal elements are bounded in magnitude by $\nu$. Therefore, by the Gershgorin circle theorem \cite{golub96}, all eigenvalues of $\D^*[i] \D[i]$ are in the range $[1 - (d-1)\nu, 1 + (d-1)\nu]$, demonstrating \eqref{eq:norm Di Di}. Furthermore, it follows that the eigenvalues of $(\D^*[i] \D[i])^{-1}$ are in the range $[(1 + (d-1)\nu)^{-1}, (1 - (d-1)\nu)^{-1}]$, leading to \eqref{eq:norm inv Di Di}.

It remains to prove \eqref{eq:norm Ds Ds}. To this end, let $|I|=\ell \le k$ and write $\D_I^* \D_I$ as
\beq
\D_I^*\D_I =
\begin{pmatrix}
\M[1,1] & \M[1,2] & \cdots & \M[1,\ell] \cr
\M[2,1] & \M[2,2] & \cdots & \M[2,\ell] \cr
\vdots  & \vdots  & \ddots & \vdots  \cr
\M[\ell,1] & \M[\ell,2] & \cdots & \M[\ell,\ell]
\end{pmatrix}
\eeq
where each $\M[i,j]$ is a $d \times d$ matrix containing the correlations between two blocks of dictionary atoms. From the definition of block coherence, we have
\beq
\|\M[i,j]\| \le d \mu_B, \quad \text{for all } i \ne j.
\eeq
By a generalization of the Gershgorin circle theorem \cite[Thm.~2]{feingold62}, it follows that all eigenvalues $\lambda$ of $\D_I^* \D_I$ satisfy
\begin{align} \label{eq:gimel 0}
\| \M[i,i] - \lambda \I \| \le \sum_{j \ne i} \|\M[i,j]\|
&\le (\ell-1)d\mu_B \notag\\
&\le (k-1)d\mu_B.
\end{align}

Now, from the definition of sub-coherence, the off-diagonal elements of $\M[i,i]$ are no larger in magnitude than $\nu$, while the diagonal elements of $\M[i,i]$ all equal $1$. Therefore, by the Gershgorin circle theorem, given an arbitrary constant $\lambda$, all eigenvalues of the $d \times d$ matrix $\M[i,i] - \lambda \I$ are in the range $[1-\lambda-(d-1)\nu, 1-\lambda+(d-1)\nu]$. Consequently
\begin{align}
\|\M[i,i] - \lambda\I\|
&\ge 1 - \lambda - (d-1)\nu.
\end{align}
Combining with \eqref{eq:gimel 0} and rearranging, we conclude that all eigenvalues of $\D_I^* \D_I$ satisfy
\beq
\lambda \ge 1 - (d-1)\nu - (k-1)d\mu_B.
\eeq
Consequently, the eigenvalues of $(\D_I^* \D_I)^{-1}$ are no larger than $(1 - (d-1)\nu - (k-1)d\mu_B)^{-1}$, establishing \eqref{eq:norm Ds Ds}.
\end{proof}

\begin{lemma} \label{le:aleph}
Consider the setting of Section~\ref{se:setting}, and suppose it is known that
\beq \label{eq:noise bound}
\max_{1 \le j \le M} \|\D^*[j] \w\|_2 < \tau
\eeq
for a given value $\tau > 0$. If the dictionary $\D$ satisfies
\beq \label{eq:maxmax cond}
\left( 1 - (d-1)\nu \right) \xmax > 2\tau + (2s-1)d \mu_B \xmax
\eeq
then
\beq \label{eq:maxmax res}
\max_{j \in S} \| \D^*[j] \y \|_2 > \max_{j \notin S} \|\D^*[j] \y \|_2
\eeq
where $S = \supp(\x)$.

If \eqref{eq:maxmax cond} is replaced by the stronger condition
\beq \label{eq:minmax cond}
\left( 1 - (d-1)\nu \right) \xmin > 2\tau + (2s-1)d \mu_B \xmax
\eeq
then
\beq \label{eq:minmax res}
\min_{j \in S} \| \D^*[j] \y \|_2 > \max_{j \notin S} \|\D^*[j] \y \|_2.
\eeq
\end{lemma}

\begin{proof}
The proof is an extension of \cite[Lemma~3]{ben-haim10} to the block-sparse case, and is ultimately inspired by \cite{donoho06}. We first note that
\begin{align} \label{eq:prf aleph 1}
&\max_{j \notin S} \|\D^*[j]\y\|_2
=   \max_{j \notin S} \left\| \D^*[j]\w + \sum_{i \in S} \D^*[j] \D[i] \x[i] \right\|_2 \notag\\
&\hspace{5mm} \le \max_{j \notin S} \|\D^*[j] \w\|_2 + \max_{j \notin S} \sum_{i \in S} \|\D^*[j] \D[i]\| \, \xmax.
\end{align}
By \eqref{eq:noise bound}, the first term in \eqref{eq:prf aleph 1} is smaller than $\tau$. Together with \eqref{eq:norm Di Dj}, we obtain
\beq \label{eq:prf aleph 5a}
\max_{j \notin S} \|\D^*[j]\y\|_2 < \tau + s d \mu_B \xmax \le \tau + k d \mu_B \xmax.
\eeq

On the other hand,
\begin{align} \label{eq:prf aleph 6}
&\max_{j \in S} \|\D^*[j] \y\|_2
= \max_{j \in S} \left\| \D^*[j]\w + \sum_{i \in S} \D^*[j]\D[i]\x[i] \right\|_2 \notag\\
&\hspace{5mm} \ge \max_{j \in S} \|\D^*[j] \D[j] \x[j]\|_2 \notag\\
&\hspace{10mm} - \max_{j \in S} \left\| \D^*[j]\w + \sum_{i \in S \backslash \{j\}} \D^*[j]\D[i]\x[i] \right\|_2.
\end{align}
As we have seen in the proof of Lemma~\ref{le:matrix norms}, the eigenvalues of $\D^*[j] \D[j]$ are bounded in the range $[1 - (d-1)\nu, 1 + (d-1)\nu]$. Consequently
\begin{align}
\max_{j \in S} \|\D^*[j] \D[j] \x[j]\|_2
&\ge \max_{j \in S} (1 - (d-1)\nu)\|\x[j]\|_2 \notag\\
&= (1 - (d-1)\nu) \xmax.
\end{align}
Combining this result with \eqref{eq:prf aleph 6}, we have
\begin{align}
&\max_{j \in S} \|\D^*[j] \y\|_2 \ge (1 - (d-1)\nu) \xmax \notag\\
&\hspace{10mm} - \max_{j \in S} \sum_{i \in S \backslash \{j\}} \| \D^*[j] \D[i] \x[i] \|_2 - \max_{j \in S} \|\D^*[j] \w\|_2.
\end{align}
Together with \eqref{eq:noise bound} and \eqref{eq:norm Di Dj}, this implies that
\begin{align} \label{eq:prf aleph 8}
&\max_{j \in S} \|\D^*[j] \y\|_2 \notag\\
&\hspace{5mm} > (1 - (d-1)\nu) \xmax - (k-1)\xmax d \mu_B - \tau \notag\\
&\hspace{5mm} = (1 - (d-1)\nu) \xmax - (2k-1)\xmax d \mu_B - 2\tau \notag\\
&\hspace{10mm}  + k \xmax d \mu_B + \tau.
\end{align}
Merging the results \eqref{eq:prf aleph 5a} and \eqref{eq:prf aleph 8} yields
\begin{align}
&\max_{j \in S} \|\D^*[j] \y\|_2
> \max_{j \notin S} \|\D^*[j]\y\|_2 \notag\\
&\hspace{5mm} + (1 - (d-1)\nu) \xmax - (2k-1)\xmax d \mu_B - 2\tau.
\end{align}
Consequently, if \eqref{eq:maxmax cond} holds, then \eqref{eq:maxmax res} follows, as required.

In a similar fashion, observe that
\begin{align} \label{eq:prf aleph 8a}
&\min_{j \in S} \|\D^*[j]\y\|_2
= \min_{j \in S} \left\| \sum_{i \in S} \D^*[j]\D[i]\x[i] + \D^*[j]\w \right\|_2 \notag\\
&\hspace{5mm} \ge \min_{j \in S} \|\D^*[j]\D[j]\x[j]\|_2 \notag\\
&\hspace{10mm}   - \max_{j \in S} \sum_{i \in S \backslash \{j\}} \|\D^*[j]\D[i]\x[i]\|_2 - \|\D^*[j]\w\|_2.
\end{align}
As noted previously, all eigenvalues of $\D^*[j] \D[j]$ are larger than or equal to $1 - (d-1)\nu$, and therefore
\beq \label{eq:prf aleph 9}
\min_{j \in S} \|\D^*[j]\D[j]\x[j]\|_2 \ge (1 - (d-1)\nu) \xmin.
\eeq
Furthermore, using \eqref{eq:norm Di Dj} we have, for $i \ne j$,
\beq \label{eq:prf aleph 10}
\|\D^*[j] \D[i] \x[i]\|_2 \le \|\D^*[j]\D[i]\| \, \xmax \le d \mu_B \xmax.
\eeq
Substituting \eqref{eq:noise bound}, \eqref{eq:prf aleph 9}, and \eqref{eq:prf aleph 10} into \eqref{eq:prf aleph 8a} provides us with
\begin{align}
&\min_{j \in S} \|\D^*[j]\y\|_2 \notag\\
&\hspace{5mm} > (1 - (d-1)\nu)\xmin - (k-1)d\mu_B \xmax - \tau \notag\\
&\hspace{5mm} = (1 - (d-1)\nu)\xmin - (2k-1)d\mu_B \xmax - 2\tau \notag\\
&\hspace{10mm}  + k d \mu_B \xmax + \tau.
\end{align}
Finally, using \eqref{eq:prf aleph 5a} we obtain
\begin{align}
&\min_{j \in S} \|\D^*[j]\y\|_2 > \max_{j \notin S} \|\D^*[j]\y\|_2 \notag\\
&\hspace{5mm} + (1 - (d-1)\nu)\xmin - (2k-1)d\mu_B \xmax - 2\tau.
\end{align}
Therefore, if the condition \eqref{eq:minmax cond} is satisfied, then \eqref{eq:minmax res} holds, completing the proof.
\end{proof}

We are now ready to prove Theorems \ref{th:bth adv} and \ref{th:bomp adv}.

\begin{proof}[Proof of Theorem~\ref{th:bth adv}]
Using \eqref{eq:adv noise} and \eqref{eq:norm Di Di}, we have for all~$j$
\begin{align}
\|\D^*[j] \w\|_2 \le \|\D[j]\| \cdot \|\w\|_2 \le \eps \sqrt{1 + (d-1)\nu} .
\end{align}
Thus, \eqref{eq:noise bound} holds with $\tau = \eps \sqrt{1 + (d-1)\nu}$.

In light of \eqref{eq:bth adv cond}, the condition \eqref{eq:minmax cond} for the second part of Lemma~\ref{le:aleph} holds, and therefore, by Lemma~\ref{le:aleph}, we conclude that \eqref{eq:minmax res} holds. It follows that all blocks $\D[i]$ with $i \in S$ are more highly correlated than the off-support blocks $\D[i], i \notin S$. Thus, the estimated support $\hS$ contains the true support set $S$ (with the possible addition of superfluous indices if $s < k$). It follows from the definition \eqref{eq:def xbth} of $\xbth$ that $(\xbth)_\hS = D_\hS^\pinv \y$, and thus
\begin{align} \label{eq:gimel 1}
&\|\x - \xbth\|_2^2
= \|\x_\hS - (\xbth)_\hS\|_2^2 \notag\\
&\hspace{5mm} =   \|\D_\hS^\pinv \D_\hS \x_\hS - \D_\hS^\pinv \y \|_2^2 \notag\\
&\hspace{5mm} \le \|\D_\hS^\pinv \|^2 \cdot \| \y - \D_\hS\x \|_2^2 \notag\\
&\hspace{5mm} =   \|\D_\hS^\pinv \|^2 \cdot \| \w \|_2^2
\end{align}
where we have used the fact that $\D_\hS^\pinv \D_\hS = \I$, which follows from our assumption that $\D_I$ has full row rank for any set $I$ of size $s$ (see Section~\ref{se:setting}).

Since $\xmin \le \xmax$, it follows from \eqref{eq:bth adv cond} that
\beq
1 - (d-1)\nu > (2k-1)d\mu_B.
\eeq Therefore, we may apply \eqref{eq:norm Ds Ds}, yielding
\begin{align} \label{eq:prf bth adv 3}
\|\D_\hS^\pinv \|^2
&= \|(\D_S^* \D_S)^{-1}\| \notag\\
&\le \frac{1}{1 - (d-1)\nu - (k-1)d\mu_B}.
\end{align}
Combining this result with \eqref{eq:gimel 1} and using \eqref{eq:adv noise}, we obtain \eqref{eq:bth adv}, as required.
\end{proof}

\begin{proof}[Proof of Theorem~\ref{th:bomp adv}]
As shown in the proof of Theorem~\ref{th:bth adv}, it follows from \eqref{eq:adv noise} that \eqref{eq:noise bound} holds with $\tau = \eps \sqrt{1 + (d-1)\nu}$. From \eqref{eq:bomp adv cond} we then have
\beq
(1 - (d-1)\nu) \xmin > 2\tau + (2k-1) d \mu_B \xmin.
\eeq
Since $\xmax \ge \xmin$, this implies the condition \eqref{eq:maxmax cond} for the first part of Lemma~\ref{le:aleph}. Thus, by Lemma~\ref{le:aleph}, the dictionary block most highly correlated with $\y$ is a block within the support $S$ of $\x$. In other words, the first iteration in the BOMP algorithm correctly identifies an element within the support $S$.

The proof continues by induction. Assume we have reached the $\ell$th iteration with $2 \le \ell \le s$ and that all previous iterations have correctly identified elements of $S$. In other words, using the notation of Section~\ref{se:algorithms}, we have $i_1, \ldots, i_{\ell-1} \in S$.

By definition, we now have
\beq \label{eq:prf bth adv 1}
\r^\ell = \y - \D \x^{\ell-1} = \D \tx^{\ell-1} + \w
\eeq
where $\tx^{\ell-1} \triangleq \x - \x^{\ell-1}$ is the estimation error after $\ell-1$ iterations. Since $\supp(\x) = S$ and, by induction, $\supp(\x^{\ell-1}) \subset S$, we have $\supp(\tx^{\ell-1}) \subset S$. Furthermore, $\ell-1 < s$, so that $\supp(\x^{\ell-1})$ contains less than $s$ elements, and is thus a strict subset of $S$. It follows that at least one nonzero block in $\tx^{\ell-1}$ is equal to the corresponding block in $\x$. Therefore
\beq \label{eq:prf bth adv 2}
\max_j \| \tx^{\ell-1}[j] \|_2 \ge \xmin.
\eeq

To summarize, by \eqref{eq:prf bth adv 1}, $\r^\ell$ can be thought of as a noisy measurement of the block sparse vector $\tx^{\ell-1}$, which contains a block whose norm is at least $\xmin$. Using \eqref{eq:prf bth adv 2} and \eqref{eq:bomp adv cond}, we find that the condition \eqref{eq:maxmax cond} holds for this modified estimation problem. Consequently, by Lemma~\ref{le:aleph}, we have
\beq
\max_{j \in S} \|\D^*[j] \r^{\ell-1}\|_2 > \max_{j \notin S} \|\D^*[j] \r^{\ell-1}\|_2.
\eeq
Therefore, by \eqref{eq:bomp i ell}, the $\ell$th iteration of the BOMP algorithm will choose an index $i_\ell$ belonging to the correct support set $S$, as long as $\ell \le s$.

Since the BOMP algorithm never chooses the same support element twice, we conclude that precisely the $s$ elements of $S$ will be identified in the first $s$ iterations. If $s<k$, then the remaining iterations will identify some additional elements not in $S$, so that ultimately the estimated support set $\hS = \{ i_1, \ldots, i_k \}$ will satisfy $\hS \supseteq S$. The estimate $\xbomp$ therefore satisfies $(\xbomp)_\hS = \D_\hS^\pinv \y$. Following the procedure \eqref{eq:gimel 1}--\eqref{eq:prf bth adv 3} in the proof of Theorem~\ref{th:bth adv}, we obtain in an identical manner the required result \eqref{eq:bomp adv}.
\end{proof}

%%%%%%%%%%%%%%%%%%%%%%%%%%%%%%%%%%%%%%%%%%%%%%%%%%%%%%%%%%%%%%%%%%%%%%%%%%%%%%%%%%%%%%%%%%%%%%%%%%%%%%%%%%%%%%%%
\section{Proof of Theorem~\ref{th:crb}}
\label{ap:crb}

To compute the CRB, we must first determine the Fisher information matrix $\J(\x)$ for estimating $\x$ from $\y$ of \eqref{eq:y=Dx+w}. This can be done using a standard formula \cite[p.~85]{kay93} and yields
\beq \label{eq:J}
\J(\x) = \frac{1}{\sigma^2} \D^* \D.
\eeq

We now identify, for each $\x \in \XX$, an orthonormal basis for the feasible direction subspace, which is defined as the smallest subspace of $\CC^N$ containing all feasible directions at $\x$. To this end, denote by $\e_i$ the $i$th column of the $N \times N$ identity matrix. Consider first points $\x \in \XX$ for which $s<k$. In other words, these are parameter values whose support $S$ contains fewer than $k$ elements. For such values of $\x$, we have, for any $\eps$ and any $1 \le i \le N$,
\beq
|\supp(\x + \eps \e_i)| \le |S| + 1 < k + 1 \le k
\eeq
and therefore $\x + \eps \e_i \in \XX$ for any $\eps$ and for any $i$. Consequently, the set of feasible directions at $\x$ includes $\{ \e_1, \ldots, \e_N \}$, and the feasible direction subspace is therefore $\CC^N$ itself. Thus, for values $\x$ containing fewer than $k$ nonzero blocks, a convenient choice of a basis for the feasible direction subspace consists of the columns of the identity matrix.

Next, consider maximal-support parameter values, i.e., vectors $\x$ for which $s=k$. It is now no longer possible to add any vector $\e_i$ to $\x$ without violating the constraints. Indeed, it is not difficult to see that the only feasible directions are linear combinations of the unit vectors $\e_i$ for which $i$ belongs to one of the blocks in $S$. These unit vectors can thus be chosen as a basis for the feasible direction subspace.

Let $\U(\x)$ be a matrix whose columns comprise the chosen orthonormal basis for the feasible direction subspace at $\x$. Note that the dimensions of $\U(\x)$ change with $\x$; specifically, $\U(\x) = \I_{N \times N}$ when $|S|<k$, and $\U(\x)$ is an $N \times sd$ matrix otherwise. A necessary condition for a finite-variance $\XX$-unbiased estimator to exist at a point $\x$ is \cite[Thm.~1]{ben-haim09b}
\beq \label{eq:unb exist cond}
\Ra{\U(x)\U^*(x)} \subseteq \Ra{\U(\x)\U^*(\x) \J(\x) \U(\x)\U^*(\x)}.
\eeq
When $s<k$, we have $\U(\x)=\I$. In this case, using \eqref{eq:J}, the condition \eqref{eq:unb exist cond} becomes
\beq
\CC^N \subseteq \Ra{\J(\x)} = \Ra{\D^* \D}.
\eeq
Since the dimensions of $\D$ are $L \times N$ with $L < N$, the rank of $\D^* \D$ is at most $L$, and thus $\Ra{\D^*\D}$ cannot include the entire space $\CC^N$. We conclude that in this case, \eqref{eq:unb exist cond} does not hold, and therefore no $\XX$-unbiased estimator exists at points $\x$ for which $|S|<s$, proving part (a) of the theorem.

Let us now turn to maximal-support parameter values $\x$. As we have seen above, in this case the matrix $\U(\x)$ consists of the columns $\e_i$ for which $i$ is an element of a block within the support of $\x$. Therefore, the product $\D \U(\x)$ selects those atoms of $\D$ belonging to blocks within $S$, i.e., $\D \U(\x) = \D_S$. Using \eqref{eq:J}, this leads to
\beq \label{eq:UJU=DsDs}
\U^*(\x) \J(\x) \U(\x) =  \frac{1}{\sigma^2} \D_S^* \D_S
\eeq
which is invertible by assumption (see Section~\ref{se:setting}). It follows that the condition \eqref{eq:unb exist cond} holds for maximal-support parameters $\x$. One can therefore apply \cite[Thm.~1]{ben-haim09b}, which states that for such values of $\x$,
\beq
\MSE(\hx,\x) \ge \Tr \! \left( \U(\x) \left( \U^*(\x) \J(\x) \U(\x) \right)^\pinv \U^*(\x) \right).
\eeq
Combining with \eqref{eq:UJU=DsDs} and using the fact that $\U^*(\x) \U(\x) = \I$, we obtain \eqref{eq:crb}, proving part (b) of the theorem.

%%%%%%%%%%%%%%%%%%%%%%%%%%%%%%%%%%%%%%%%%%%%%%%%%%%%%%%%%%%%%%%%%%%%%%%%%%%%%%%%%%%%%%%%%%%%%%%%%%%%%%%%%%%%%%%%
\section{Proofs for Gaussian Noise}
\label{ap:prf gauss}

We begin with two lemmas which prove some useful properties of the Gaussian distribution. The first of these is a generalization of a result due to \v{S}id\'ak \cite{sidak67}.

\begin{lemma} \label{le:sidak gen}
Let $\bv_1, \ldots, \bv_M$ be a set of $M$ jointly Gaussian random vectors. Suppose that $\E{\bv_i} = \zero$ for all $i$, but that the covariances of the vectors are unspecified and that the vectors are not necessarily independent. We then have
\begin{align}
&\Pr{\|\bv_1\|_2 \le c_1, \|\bv_2\|_2 \le c_2, \ldots, \|\bv_M\|_2 \le c_M} \notag\\
&\hspace{5mm} \ge \Pr{\|\bv_1\|_2 \le c_1} \cdot \Pr{\|\bv_2\|_2 \le c_2} \cdots \notag\\
&\hspace{45mm} \cdots \Pr{\|\bv_M\|_2 \le c_M}.
\end{align}
\end{lemma}

\begin{proof}
We will demonstrate that
\begin{align} \label{eq:sidak 1}
&\Pr{\|\bv_1\|_2 \le c_1, \|\bv_2\|_2 \le c_2, \ldots, \|\bv_M\|_2 \le c_M} \notag\\
&\hspace{5mm} \ge \Pr{\|\bv_1\|_2 \le c_1} \Pr{\|\bv_2\|_2 \le c_2, \ldots, \|\bv_M\|_2 \le c_M}.
\end{align}
The result then follows by induction. For simplicity of notation, we will prove that \eqref{eq:sidak 1} holds for the case $M=2$; the general result can be shown in the same manner.

Denote by $f(\bv_1|\bv_2)$ the pdf of $\bv_1$ conditioned on $\bv_2$. Observe that, for a deterministic value $\w$, the pdf $f(\bv_1|\w)$ defines a Gaussian random vector whose mean depends linearly on $\w$, but whose covariance is constant in $\w$. Therefore, using a result due to Anderson \cite{anderson55}, it follows that
\beq
\Pr{ \|\bv_1\|_2 \le c_1 | \bv_2 = \alpha \w } = \int_{\|\u_1\|_2 \le c_1} f(\u_1|\alpha \w) d\u
\eeq
is a non-increasing function of $\alpha$.

Next, denoting by $f(\bv_2)$ the marginal pdf of $\bv_2$, we have
\begin{align}
a(c_1,c_2) &\triangleq \Pr{\|\bv_1\|_2 \le c_1 \big| \, \|\bv_2\|_2 \le c_2} \notag\\
%&= \frac{\Pr{\|\bv_1\|_2 \le c_1, \|\bv_2\|_2 \le c_2}}{\Pr{\|\bv_2\|_2 \le c_2}} \notag\\
&= \frac{ \int_{\|\u\|_2 \le c_1} \int_{\|\w\|_2 \le c_2} f(\u|\w) f(\w) \, d\w \, d\u}{\Pr{\|\bv_2\|_2 \le c_2}} \notag\\
&= \frac{ \int_{\|\w\|_2 \le c_2} \Pr{\|\bv_1\|_2 \le c_1 | \bv_2 = \w} f(\w) \, d\w }{ \int_{\|\w\|_2 \le c_2} f(\w) \, d\w }.
\end{align}
Thus, the function $a(c_1,c_2)$ is a weighted average of expressions of the form $\Pr{\|\bv_1\|_2 \le c_1 | \bv_2 = \w}$ for values of $\w$ satisfying $\|\w\|_2 \le c_2$. However, as we have shown, $\Pr{\|\bv_1\|_2 \le c_1 | \bv_2 = \w}$ is non-increasing in $\|\w\|_2$. Consequently, $a(c_1,c_2)$ is non-increasing in $c_2$.

On the other hand, observe that as $c_2 \rightarrow \infty$, the probability of the event $\|\bv_2\|_2 \le c_2$ converges $1$. Thus we have
\beq
\lim_{c_2 \rightarrow \infty} a(c_1,c_2) = \Pr{\|\bv_1\|_2 \le c_1}.
\eeq
Combined with the fact that $a(c_1,c_2)$ is non-increasing in $c_2$, we find that
\beq
a(c_1,c_2) \ge \Pr{\|\bv_1\|_2 \le c_1} \quad \text{for all $c_1, c_2$}.
\eeq
Using the definition of $a(c_1,c_2)$ and applying Bayes's rule, we obtain
\begin{align}
&\Pr{\|\bv_1\|_2 \le c_1, \|\bv_2\|_2 \le c_2} \notag\\
&\hspace{5mm} \ge \Pr{\|\bv_1\|_2 \le c_1} \Pr{\|\bv_2\|_2 \le c_2}
\end{align}
and thus complete the proof.
\end{proof}

Our next lemma bounds the tail probability of the chi-squared distribution.

\begin{lemma} \label{le:chi square}
Let $\u$ be a $d$-dimensional Gaussian random vector having mean zero and covariance $\I$. Then, for any $t \ge 1$, we have
\begin{subequations} \label{eq:le:chi square both}
\begin{align}
\Pr{\|\u\|_2^2 \ge t^2}
&\le \frac{(d-2)!! \lceil d/2 \rceil}{2^{d/2-1} \Gamma(d/2)} t^{d-2} e^{-t^2/2} \label{eq:le:chi square gamma} \\
&\le 0.8 d t^{d-2} e^{-t^2/2} \label{eq:le:chi square}
\end{align}
\end{subequations}
where $\Gamma(z) \triangleq \int_0^\infty t^{z-1} e^{-t} dt$ is the Gamma function and
\beq
n!! \triangleq \prod_{0 \le i < n/2} (n-2i)
\eeq
is the double factorial operator.
\end{lemma}

Of the two bounds provided in \eqref{eq:le:chi square both}, the first is somewhat tighter, but obviously more cumbersome. For analytical tractability, we will use the latter bound in the sequel.

\begin{proof}[Proof of Lemma~\ref{le:chi square}]
The expression $\|\u\|_2^2$ is distributed as a chi-squared random variable with $d$ degrees of freedom. Therefore, its tail probability is given by \cite[\S 16.3]{kendall-vol1}
\beq \label{eq:chi sq cdf}
\Pr{\|\u\|_2^2 \ge t^2} = \frac{\Gamma(d/2, t^2/2)}{\Gamma(d/2)}
\eeq
where $\Gamma(a,z)$ is the incomplete Gamma function $\Gamma(a,z) \triangleq \int_z^{\infty} t^{a-1}\,e^{-t}\,dt$. It follows from the series expansion of $\Gamma(a,z)$ that \cite[\S 6.5.32]{abramowitz64}
\begin{align} \label{eq:prf chi sq 1}
&\Gamma\!\left( \frac d 2 , \frac{t^2}{2} \right)
\le \frac{e^{-t^2/2}}{2^{d/2-1} t^2} \big[ t^d + (d-2)t^{d-2}\notag\\
&\hspace{10mm}   + (d-2)(d-4)t^{d-4}+ \cdots + (d-2)!! t^m \big]
\end{align}
where $m = 1$ when $d$ is odd and $m=2$ when $d$ is even. Note that \eqref{eq:prf chi sq 1} holds with equality for even $d$, but the inequality is strict for odd $d$. Since $t \ge 1$, we can enlarge each of the terms in the square brackets in \eqref{eq:prf chi sq 1} by replacing it with $(d-2)!! t^d$. The total number of terms in brackets is $\lceil d/2 \rceil$, yielding
\beq
\Gamma\!\left( \frac d 2 , \frac{t^2}{2} \right) \le \frac{e^{-t^2/2}}{2^{d/2-1}} t^{d-2} (d-2)!! \left\lceil \frac d 2 \right\rceil.
\eeq
Substituting into \eqref{eq:chi sq cdf} demonstrates \eqref{eq:le:chi square gamma}.

To prove \eqref{eq:le:chi square}, we distinguish between even and odd values of $d$. Assume first that $d$ is even and denote $d = 2p$. We then have
\beq
\Gamma(d/2) = \Gamma(p) = (p-1)!
\eeq
and
\beq
(d-2)!! = (2p-2)!! = 2^{p-1} (p-1)!.
\eeq
Substituting these values into \eqref{eq:le:chi square gamma} and simplifying yields
\beq
\Pr{\|\u\|_2^2 \ge t^2} \le \frac d 2  t^{d-2} e^{-t^2/2}
\eeq
which clearly satisfies \eqref{eq:le:chi square}.

Similarly, assume that $d$ is odd and write $d = 2p+1$. Substituting the formula
\beq
\Gamma(d/2) = \Gamma(p + 1/2) = \frac{(2p-1)!! \sqrt{\pi}}{2^p}
\eeq
into \eqref{eq:le:chi square gamma}, we obtain
\beq \label{eq:prf chi sq 2}
\Pr{\|\u\|_2^2 \ge t^2} \le \sqrt{\frac{2}{\pi}} \, \frac{d+1}{2} t^{d-2} e^{-t^2/2}.
\eeq
It is easily verified that
\beq
\sqrt{\frac{2}{\pi}}\, \frac{d+1}{2} \le 0.8d \quad \text{for all $d \ge 1$.}
\eeq
Substituting back into \eqref{eq:prf chi sq 2} yields the required result.
\end{proof}

Our next result applies more specifically to the block sparse estimation setting. Following \cite{tropp06, ben-haim10}, we consider the event
\beq \label{eq:def B}
B = \left\{ \max_{1 \le i \le M} \|\D^*[i] \w\|_2^2 \le \tau^2 \right\}
\eeq
where
\beq \label{eq:def tau}
\tau^2 = 2 d\sigma \alpha (1 + (d-1)\nu) \log N
\eeq
for a given $\alpha > 1/(2d \log N)$. We then have the following lemma.

\begin{lemma} \label{le:prob B}
Under the setting of Section~\ref{se:setting}, assume that $\w$ is a Gaussian random vector with mean zero and covariance $\sigma^2 \I$. Then, the probability of the event $B$ of \eqref{eq:def B} is bounded by
\beq \label{eq:le:prob B}
\Pr{B} \ge 1 - \frac{0.8 (2 \alpha d \log N)^{d/2-1}}{N^{\alpha d-1}}.
\eeq
\end{lemma}

\begin{proof}
% slightly more detailed proof: notebook p.701 top
Observe that $\D^*[i]\w$ is a $d$-dimensional Gaussian random vector with mean zero and covariance $\sigma^2 \D^*[i] \D[i]$. Therefore, the random vector
\beq
\u = \frac 1 \sigma  (\D^*[i] \D[i])^{-1/2} \D^*[i] \w
\eeq
is a $d$-dimensional Gaussian random vector with mean zero and covariance $\I$. We thus have
\begin{align}
\Pr{\|\D^*[i]\w\|_2^2 \le \tau^2}
&=   \Pr{\sigma^2\|(\D^*[i]\D[i])^{1/2}\u\|_2^2 \le \tau^2} \notag\\
&\ge \Pr{\sigma^2\|\D^*[i]\D[i]\| \cdot \|\u\|_2^2 \le \tau^2} \notag\\
&\ge \Pr{\|\u\|_2^2 \le \frac{\tau^2}{\sigma^2 (1 + (d-1)\nu)}}
\end{align}
where, in the last step, we used \eqref{eq:norm Di Di}. Using Lemma~\ref{le:chi square} and substituting the value \eqref{eq:def tau} of $\tau^2$, we obtain
\begin{align} \label{eq:prf prob B 1}
&\Pr{\|\D^*[i]\w\|_2^2 \le \tau^2} \ge 1 - \eta
\end{align}
where
\begin{align} \label{eq:eta}
\eta &\triangleq 1 - 0.8 d (2 \alpha d \log N)^{d/2-1} \exp(-d\alpha \log N) \notag\\
&= 1 - \frac{ 0.8 d (2 \alpha d \log N)^{d/2-1} }{N^{\alpha d}}.
\end{align}

Using Lemma~\ref{le:sidak gen}, we have
\begin{align}
\Pr{B} &\ge \prod_{i=1}^M \Pr{ \|\D^*[i]\w\|_2^2 \le \tau^2 } \notag\\
&= (1 - \eta)^M.
\end{align}
When $\eta > 1$, the bound \eqref{eq:le:prob B} is meaningless and the theorem holds vacuously. Otherwise, when $\eta \le 1$, we have
\beq
\Pr{B} \ge 1 - M \eta
\eeq
where we used the fact that $(1-\eta)^M \ge 1-M\eta$ whenever $\eta \le 1$ and $M \ge 1$. Substituting the value of $\eta$ from \eqref{eq:eta} and recalling that $N = Md$ yields the required result.
\end{proof}

We are now ready to prove Theorems \ref{th:bth gauss} and \ref{th:bomp gauss}.

\begin{proof}[Proof of Theorem~\ref{th:bth gauss}]
By Lemma~\ref{le:prob B}, the event $B$ of \eqref{eq:def B} occurs with probability exceeding \eqref{eq:bth gauss prob}. Furthermore, using \eqref{eq:bth gauss cond}, it follows from Lemma~\ref{le:aleph} that under the event $B$, all blocks in the correct support set $S$ are more highly correlated with $\y$ than the off-support blocks. Consequently, when $B$ occurs, we have $S \subseteq \hS$, where $\hS$ is the support estimated by the BTH algorithm. Note, however, that the estimated set $\hS$ will contain additional blocks not in $S$ if $s<k$. It follows that
\begin{align} \label{eq:prf th:bth gauss 1}
&\|\x - \xbth\|_2^2
= \|\x_\hS - (\xbth)_\hS\|_2^2 \notag\\
&\hspace{5mm} =   \|\D_\hS^\pinv \D_\hS \x_\hS - \D_\hS^\pinv \y \|_2^2 \notag\\
&\hspace{5mm} \le \|(\D_\hS^* \D_\hS)^{-1}\|^2 \cdot \| \D_\hS^*\w \|_2^2 \notag\\
&\hspace{5mm} \le \|(\D_\hS^* \D_\hS)^{-1}\|^2 \cdot \sum_{i \in \hS} \| \D^*[i]\w \|_2^2
\end{align}
where we have used the fact that $\D_\hS^\pinv \D_\hS = \I$, which is a consequence of the assumption that $\D_\hS$ has full row rank (see Section~\ref{se:setting}). Using \eqref{eq:norm Ds Ds} and \eqref{eq:def B}, we have that when $B$ occurs
\begin{align} \label{eq:prf th:bth gauss 2}
&\|\x - \xbth\|_2^2
\le \frac{k \tau^2}{(1 - (d-1)\nu - (k-1)d\mu_B)^2}.
\end{align}
Substituting the value \eqref{eq:def tau} of $\tau$ yields the required result \eqref{eq:bth gauss perf}.
\end{proof}

\begin{proof}[Proof of Theorem~\ref{th:bomp gauss}]
It follows from Lemma~\ref{le:prob B} that the event $B$ occurs with probability exceeding \eqref{eq:bth gauss prob}. Our goal in this proof will thus be to show that, if $B$ does occur, then the BOMP algorithm correctly identifies all elements of the support $S$ of $\x$ (although some off-support elements may be identified as well if $s<k$). The remainder of the proof will then follow the steps of the proof of Theorem~\ref{th:bth gauss}.

To demonstrate that the correct support is recovered, we begin by analyzing the first iteration of the BOMP algorithm. This iteration chooses a block $i_1$ having maximal correlation $\|\D^*[i_1] \y\|_2$ with the measurements $\y$. Now, since $\xmax \ge \xmin$, the condition \eqref{eq:bomp gauss cond} implies \eqref{eq:maxmax cond}, with $\tau$ given by \eqref{eq:def tau}. Consequently, by Lemma~\ref{le:aleph}, under the event $B$ we find that the first iteration of BOMP identifies an element $i_1$ in the correct support set $S$.

To show that the next $s-1$ iterations of the BOMP algorithm also identify support elements, we proceed by induction. Specifically, assume that $\ell-1 < s$ iterations have correctly identified elements $i_1,\ldots,i_{\ell-1}$, all of which are in the support set $S$. As in the proof of Theorem~\ref{th:bomp adv}, define the estimation error after $\ell-1$ iterations as $\tx^{\ell-1} \triangleq x - x^{\ell-1}$. By the induction hypothesis, $\supp(\tx) \subset S$, and clearly $\supp(\x) = S$. Thus $\supp(\tx) \subset S$, i.e., the support of $\tx$ is a strict subset of $S$. Using the same arguments as in the proof of Theorem~\ref{th:bomp adv}, we find that $\tx^{\ell-1}$ contains a block whose norm is at least $\xmin$. Therefore, we can consider a modified estimation problem, in which $\r^{\ell}$ is a noisy measurement vector of the block sparse signal $\tx^{\ell-1}$. Together with \eqref{eq:bomp gauss cond}, this implies that \eqref{eq:maxmax cond} holds for the modified setting. Therefore, by \eqref{eq:maxmax res}, the block in $\r^\ell$ having maximal correlation with the measurements is an element of $S$. Consequently, BOMP will correctly identify a support element in the $\ell$th iteration. Since the BOMP algorithm never selects a previously chosen support element, we find by induction that the support set $S$ will be identified in full after $s$ iterations. If $s<k$, then the remaining $k-s$ iterations will identify arbitrary off-support elements.

Denoting by $\hS$ the complete $k$-element support set identified by the BOMP approach, we thus have $S \subseteq \hS$. Following the technique \eqref{eq:prf th:bth gauss 1}--\eqref{eq:prf th:bth gauss 2} used in the proof of Theorem~\ref{th:bth gauss} thus yields the required result \eqref{eq:bomp gauss perf}.
\end{proof}

\bibliographystyle{IEEEtran}
\bibliography{IEEEabrv,zvika}

\end{document}